\documentclass[11pt]{article}
\RequirePackage[OT1]{fontenc}
\RequirePackage{float}
\usepackage{amsmath}
\usepackage{amssymb}
\usepackage{bm}
\usepackage{bbm}
\usepackage{mathtools}
\usepackage{setspace}
\usepackage[margin=1in]{geometry}
\usepackage[utf8]{inputenc}
\usepackage[english]{babel}
\usepackage{mathrsfs}
\usepackage{subcaption}
\usepackage{amsthm}
\usepackage{afterpage}
\usepackage[authoryear]{natbib}
\usepackage{authblk}
\usepackage{multirow}
\usepackage{graphicx}
\usepackage{algorithm}
\usepackage{algorithmic}
\usepackage[colorlinks]{hyperref} 

\newtheorem{theorem}{Theorem}[section]

\newtheorem{lemma}[theorem]{Lemma}

\newcommand{\bea}{\begin{eqnarray*}}
\newcommand{\eea}{\end{eqnarray*}}
\newcommand{\bean}{\begin{eqnarray}}
\newcommand{\eean}{\end{eqnarray}}

\newcommand{\lra}{\longrightarrow}

\newcommand{\sg}{\Sigma}
\newcommand{\what}{\widehat}

\newcommand{\bbR}{\mathbb{R}}
\newcommand{\bbE}{\mathbb{E}}

\begin{document}
\doublespacing
\title{\large \textbf{Bayesian Group Selection in Logistic Regression with Application to MRI Data Analysis}}
\author[1]{Kyoungjae Lee\footnote{Corresponding author.}}
\author[2]{Xuan Cao}
\affil[1]{Department of Statistics, Inha University}
\affil[2]{Department of Mathematical Sciences, University of Cincinnati}

\maketitle
\begin{abstract}
	We consider Bayesian logistic regression models with group-structured covariates.
	In high-dimensional  settings, it is often assumed that only small portion of groups are significant, thus consistent group selection is of significant importance.
	While consistent frequentist group selection methods have been proposed, theoretical properties of Bayesian group selection methods for logistic regression models have not been investigated yet.
	In this paper, we consider a hierarchical group spike and slab prior for logistic regression models in high-dimensional settings.
	Under mild conditions, we establish strong group selection consistency of the induced posterior, which is the first theoretical result in the Bayesian literature.
	Through simulation studies, we demonstrate that the performance of the proposed method outperforms existing state-of-the-art methods in various settings.
	We further apply our method to an MRI data set for predicting Parkinson's disease and show its benefits over other contenders.
\end{abstract}


\section{Introduction} \label{sec:introduction}
Consider a standard logistic regression model 
\begin{align*}
P\left(E_i = 1 \mid x_i,\beta_{1},\ldots,\beta \right) = \frac{\exp\left\{ x_{i}^T\beta \right\}}{1+\exp\left\{ x_{i}^T\beta \right\}}, \quad i = 1,2, \ldots, n,
\end{align*}
where $E_i \in \{0,1\}$ is a binary response, $x_i\in\bbR^p$ is a vector of covariates and $\beta \in \bbR^p$ is a vector of coefficients.
When the number of covariates, $p$, is much larger than the sample size, $n$, variable selection is one of the most important tasks to uncover significant covariates which actually affect a response.
Nowadays, such high-dimensional data sets, including Magnetic Resonance Imaging (MRI), climate and gene expression data, are very common in various areas.

In many applications, a group structure of covariates often arises naturally, in which the $p$-dimensional coefficient vector $\beta$ can be divided into many smaller vectors.
Once a group structure of covariates is given, it is assumed that the coefficients in the same group are all zeros or all nonzeros.
For example, in a regression model with a multi-level categorical variable, a group of dummy variables encodes the same factor, thus their coefficients are all zeros (nonzeros) if the corresponding categorical variable is insignificant (significant).
For more examples for applications with group-structured covariates, we refer \cite{Xu:Ghosh:2015} and \cite{Yang:Naveen:2018}.
Given a group structure, variable  selection problem becomes a {\it group selection} problem to detect significant groups, so a tailored variable selection method incorporating the group information is required for efficient inference \citep{breheny2009penalized}.
Following the setup in \cite{Xu:Ghosh:2015} and \cite{Wen:2019}, we will assume that the covariates are arranged according to a known group information that are formed naturally.


In the literature, various statistical methods for logistic regression models with group-structured covariates have been developed.
\cite{meier2008group} proposed the group lasso for logistic regression, which is shown to be consistent even in high-dimensional settings.
However, since lasso-type methods in general require the irrepresentable condition for the design matrix to consistently recover the true support \citep{wainwright2009sharp}, several methods were developed in an effort to improve the selection performance including the group smoothly clipped absolute deviation (SCAD) \citep{wang2007group}, group minimax concave penalty (MCP) \citep{breheny2009penalized} and group exponential lasso \citep{breheny2015group}. 
Recently, \cite{wang2019adaptive} investigated theoretical properties of the adaptive group lasso for generalized linear model (GLM) in high-dimensional settings, and showed that the adaptive group lasso enjoys selection consistency under certain regularity conditions.

On the Bayesian side, \cite{Xu:Ghosh:2015} suggested the Bayesian group lasso with spike and slab priors in linear regression settings.
Under the orthogonal design, they showed variable selection consistency of the posterior median.
\cite{Yang:Naveen:2018} considered a similar Bayesian hierarchical model for group selection in linear regression settings and  obtained strong selection consistency.
However, both \cite{Xu:Ghosh:2015} and \cite{Yang:Naveen:2018} focused on linear regression models, and their theoretical results are not directly applicable to logistic regression models.
Recently, \cite{Naveen:2018} proposed the skinny Gibbs sampler for high-dimensional logistic regression models and established strong selection consistency, but their method is not suitable when there is a group structure between covariates.
Furthermore, they showed selection consistency of the {\it approximate} posterior, called the skinny Gibbs posterior, rather than the {\it true} posterior.
\cite{tang2018group} and \cite{Wen:2019} proposed the group spike and slab lasso \citep{rovckova2018spike} and the grouped automatic relevance determination priors for logistic regression models, respectively, but theoretical properties of posteriors were not investigated.
To the best of our knowledge, strong selection consistency of posteriors for group selection in logistic regression has not been established yet.

In this paper, we propose a Bayesian hierarchical model for group selection in high-dimensional logistic regression models.
To select the groups of significant covariates, we use the group spike and slab prior consisting of a point mass spike and a multivariate normal slab components.
Under regularity conditions, we establish posterior ratio consistency (Theorem \ref{thm:ratio}), which guarantees that the truely significant groups, say the {\it true model}, will be the mode of our posterior with large probability tending to 1.
We further show that, under similar conditions, our posterior enjoys strong selection consistency (Theorem \ref{thm:selection}), which implies that the posterior probability assigned to the true model converges in probability to 1 as we observe more data.
For posterior computation, we suggest two algorithms, the standard Gibbs sampler and its variant by introducing the neuronized prior form \citep{Shin:Liu:2018}, and investigate their performances in various simulation settings.
In simulation studies, we demonstrate that the proposed Bayesian method in this paper can outperform other state-of-the-art statistical methods.
Finally, the proposed method is applied to an MRI data set for predicting Parkinson's disease. 
Specifically, we divide the voxels from the three-dimensional MRI image data into several groups according to prior structural information (brain atlases), and the voxels within the same brain region are automatically organized into a single group. 
We attempt to select a small proportion of these groups for diagnosis of Parkinson's disease, and it turns out that our method has better prediction performance compared with other contenders.

The rest of the paper is organized as follows.
In Section \ref{sec:model spec}, the group spike and slab prior for a logistic regression model is described.
We show strong selection consistency results in Section \ref{sec:selection cons}, where proofs are provided in Supplementary material.
Posterior computation algorithms are described in Section \ref{sec:post_comp}, and we show the performance of the proposed method and compare it with other competitors through simulation studies in Section \ref{sec:sim}.
In Section \ref{sec:real}, a real data analysis is conducted for predicting Parkinson's disease, and a discussion is given in Section \ref{sec:disc}.

\section{Preliminaries} \label{sec:model spec}

\subsection{Notation}
For any $a, b \in \bbR$, $a\vee b$ and $a \wedge b$ mean the maximum and minimum of $a$ and $b$, respectively.
For any sequences $a_n$ and $b_n$, we denote $a_n \lesssim b_n$, or equivalently $a_n =O(b_n)$, if there exists a constant $C>0$ such that $|a_n| \le C |b_n|$ for all large $n$.
We denote $a_n \ll b_n$, or equivalently $a_n =o(b_n)$, if $a_n / b_n \lra 0$ as $n\to\infty$.
Without loss of generality, if $a_n \ge b_n>0$ and there exist constants $C_1>C_2>0$ such that $C_2 < b_n/ a_n \le a_n / b_n < C_1$, we denote $a_n \sim b_n$.
For a given vector $v = (v_1,\ldots, v_p)^T \in \bbR^p$, the vector $\ell_2$-norm is denoted as $\|v\|_2 = ( \sum_{j=1}^p v_j^2 )^{1/2}$.
For any real symmetric matrix $A$, $\lambda_{\max}(A)$ and $\lambda_{\min}(A)$ are the maximum and minimum eigenvalue of $A$, respectively.
For a given matrix $X \in \bbR^{n \times p}$, let $X_S \in \bbR^{n\times |S|}$ denote the submatrix of $X$ containing the columns indexed by $S \subseteq [p] =: \{1,\ldots, p \}$, where $|S|$ is the cardinality of $S$.

\subsection{Model Specification}
We first describe the framework for Bayesian group selection followed by our hierarchical model specification. 
Let $E \in \{0,1\}^{n}$ be the binary response vector and $X \in \bbR^{n\times p}$ be the design matrix. 
Without loss of generality, we assume that the columns of $X$ are standardized to have zero mean and unit variance. 
Let $x_i \in \bbR^p$ denote the $i$th row vector of $X$ that contains the covariates for the $i$th subject. 
We consider a standard logistic regression model with $r$ groups:
\begin{align} \label{logistic_model}
P\left(E_i = 1 \mid x_i,\beta_{G_1},\ldots,\beta_{G_r} \right) = \frac{\exp\left\{\sum_{j=1}^r x_{iG_j}^T\beta_{G_j}\right\}}{1+\exp\left\{\sum_{j=1}^r x_{iG_j}^T\beta_{G_j}\right\}}, \quad i = 1,2, \ldots, n,
\end{align}
where, for $j = 1, 2, \ldots, r$, $\beta_{G_j} \in \bbR^{|G_j|}$ is a coefficients vector, $X_{G_j} \in \bbR^{n\times |G_j|}$ is the submatrix of $X$ corresponding to the $j$th group, and $x_{iG_j} \in \bbR^{|G_j|}$ is the $i$th row of matrix $X_{G_j}$ that contains the covariates within the $j$th group for the $i$th subject. 
Note that $x_i^T = \left(x_{iG_1}^T, x_{iG_2}^T, \ldots, x_{iG_r}^T\right)$ and $p = \sum_{j =1}^r|G_j|$.
We will work in a scenario where the dimension of predictors, $p$, as well as the number of groups, $r$, grows with the sample size $n$. 
Thus, we consider the number of predictors and groups are functions of $n$, that is, $p = p_n$ and $r= r_n$, but we denote them as $p$ and $r$ for notational simplicity.

Our goal is group selection, i.e., to correctly identify all the nonzero group coefficients. In light of that, for a given integer $1 \le j \le r$, we introduce binary latent variable $Z_j$, for the $j$th group coefficient vector $\beta_{G_j}$, which indicates whether the $j$th group covariates are active (i.e., having nonzero coefficients). We consider the following hierarchical group spike and slab priors in \cite{Yang:Naveen:2018} and \cite{Xu:Ghosh:2015} over the grouped regression coefficients, for $1 \le j \le r$,
\begin{eqnarray} \label{basad}
&\beta_{G_j} \mid Z_j \sim (1 - Z_j) \delta_0 + Z_j  N_{|G_j|}( 0, \tau^2I_{|G_j|}),    \label{model1}\\
&P(Z_j = 1) = 1 - P(Z_j = 0) = q_j, \label{model2}
\end{eqnarray}
for some $\tau>0$ and $0<q_j<1$, where $\delta_0$ is a point mass at zero.
Here, $Z_j= 1$ implies $\beta_j$ being the “signal” (i.e., from the slab component), and $Z_j = 0$ implies $\beta_{G_j}$ being the ``noise'' (i.e., from the spike component). 
Model \eqref{model2} imposes independent Bernoulli priors over $Z_j$ for $1 \le j \le r$. 
The proposed hierarchical model now has two hyperparameters: the scale parameter $\tau$ in model (\ref{model1}) controlling the variance of the slab part, the inclusion probability $q_j$ in model \eqref{model2} penalizing large models to encourage the true group coefficients to be sparse.
To obtain our desired consistency results, appropriate conditions for these hyperparameters will be introduced in Section \ref{sec:selection cons} (Condtion  \hyperref[cond_A4]{\rm (A4)}).

The intuition behind this set-up with latent variables is that a zero or very small group coefficient vector $\beta_{G_j}$ will be identified with $Z_j=0$, while an active coefficient vector will be classified as $Z_j = 1$. 
We use the posterior probabilities of all the $r$ latent variables to identify the active group coefficients.

\section{Main results}\label{sec:selection cons}

In this section, we show that the proposed Bayesian model enjoys desirable theoretical properties.
Let $t \subseteq [r]$ be the true model, which means that the nonzero locations of the true coefficient vector are $G_t = (G_j , j \in t)$.
Let $\beta_0 \in \bbR^p$ be the true coefficient vector and $\beta_{0, G_t} \in \bbR^{|G_t|}$ be the vector of the true nonzero coefficients.
For a given model $k \subseteq [r]$, we denote $L_n(\beta_{G_k})$ and $s_n(\beta_{G_k}) = \partial L_n(\beta_{G_k}) /(\partial \beta_{G_k})$ as the log likelihood and 
score function, respectively.
Furthermore,
\bea
H_n( \beta_{G_k}) &=& - \frac{\partial^2 L_n(\beta_{G_k}) }{\partial \beta_{G_k} \partial \beta_{G_k}^T} \,\,=\,\, \sum_{i=1}^n \sigma_i^2(\beta_{G_k}) x_{i G_k} x_{i G_k}^T \,\,=\,\, X_{G_k}^T \sg(\beta_{G_k}) X_{G_k}
\eea
as the negative Hessian of $L_n(\beta_{G_k})$, where $\sg(\beta_{G_k}) \equiv \sg_{G_k} = diag(\sigma_1^2(\beta_{G_k}),\ldots, \sigma_n^2(\beta_{G_k}))$, $\sigma_i^2(\beta_{G_k}) = \mu_i(\beta_{G_k})(1- \mu_i(\beta_{G_k}))$ and 
\bea
\mu_i(\beta_{G_k}) &=& \frac{\exp \{x_{i G_k}^T \beta_{G_k} \} }{1+ \exp \{x_{i G_k}^T \beta_{G_k} \}}.
\eea
In the rest of the paper, we denote $\sg = \sg (\beta_{G_t})$ and $\sigma_i^2 = \sigma_i^2(\beta_{G_t})$ for simplicity.

To attain desirable asymptotic properties of our posterior, we assume the following conditions: \\
\noindent{\bf Condition (A1)}\label{cond_A1} $r = r_n \lra \infty$ and $\log n \lesssim \log r =o(n)$ as $n\to \infty$.\\
\noindent{\bf Condition (A2)}\label{cond_A2} For some constant $C \in (0,\infty)$ and $0 \le d < (1+d)/2 \le d' \le 1$,
\bea
\max_{i,j} |x_{ij}| &\le& C , \\
0 < \lambda  \le \min_{k : |G_k| \le m_n + |G_t|} \lambda_{\min}\Big( n^{-1}H_n(\beta_{0, G_k}) \Big) 
&\le& \Lambda_{m_n + |G_t|} \le C^2 \Big( \frac{n}{\log r} \Big)^d,
\eea 
and $|G_t| \le   m_n$, where  $m_n = \big\{  (  n /\log r)^{\frac{1-d'}{2}}  \wedge p\big\}$ and $\Lambda_{\zeta} = \max_{k : |k| \le \zeta} \lambda_{\max} ( n^{-1}  X_{G_k}^T X_{G_k}  )$ for any integer $\zeta >0$. 
Furthermore, $\|\beta_{0,G_t}\|_2^2 = O\big( \tau^2 (\log r/n)^d \big)$.\\
\noindent{\bf Condition (A3)}\label{cond_A3} 
For some  constant $c_0>0$,
\bean\label{beta-min}
\min_{j \in t}  \|\beta_{0, G_j}\|_2^2 &\ge& c_0 \frac{|G_t| \Lambda_{|t|} \log r }{n} .
\eean
\noindent{\bf Condition (A4)}\label{cond_A4} For some small constant $\delta > 0$, the hyperparameters satisfy 
\bea
 \tau^2 \sim (1 \vee n^{-1} r^{2+2\delta} )  &\text{ and }&  q_j \equiv q \sim r^{-1}.
\eea

In Condition \hyperref[cond_A1]{\rm (A1)}, $\log n \lesssim \log r$ means that the number of groups, $r$, increases at a certain rate.
If $r \ge n^c$ for some small constant $c>0$, this condition is satisfied.
Condition \hyperref[cond_A1]{\rm (A1)} also allows $r$ to grow $\exp \{  o(n)\}$ as $n\to\infty$, which is the strongest result in the literature \citep{Yang:Naveen:2018}.

Condition \hyperref[cond_A2]{\rm (A2)} gives  lower and upper bounds of $n^{-1}H_n( \beta_{0,G_k} )$ and $n^{-1} X_{G_k}^T X_{G_k}$, respectively, where $k$ is a large model satisfying $|G_k| \le m_n + |G_t|$.
The lower bound condition can be seen as a restricted eigenvalue condition for $k$-sparse vectors and is satisfied with high probability for sub-Gaussian design matrices \citep{Naveen:2018}.
Similar conditions have been used in the linear regression literature \citep{ishwaran2005spike,yang2016computational,song2017nearly}. 
The upper bound condition is much weaker than or comparable to the (bounded) maximum eigenvalue conditions assumed in \cite{bondell2012consistent}, \cite{johnson2012bayesian}, \cite{narisetty2014bayesian} and \cite{CKG:2019}.

The quantity $m_n$ in Condition \hyperref[cond_A2]{\rm (A2)} is essentially an effective dimension of our model. 
We will restrict the prior on $Z$ to have the size smaller than or equal to $m_n$, which means that we focus on a set of {\it reasonably large} models.
The condition $|G_t|\le m_n$ implies that the true model is always included in this set.
Similar conditions have been commonly assumed in the sparse estimation literature \citep{Naveen:2018, LLL:2019}.

The last assumption in Condition \hyperref[cond_A2]{\rm (A2)} says that the magnitude of true signals is bounded above $\tau^2 (\log r /n)^d$ up to some constant.
In high-dimensional settings $r \ge n$, together with Condition \hyperref[cond_A4]{\rm (A4)}, the upper bound grows polynomially along with $r$, which allows the magnitude of signals to increase to infinity.
It is known that this upper bound condition can be avoided by using heavier tailed priors \citep{castillo2015bayesian} or using data-dependent priors with sub-Gaussian tail \citep{martin2017empirical} for the slab part in \eqref{model1}.
However, in this paper, we proceed with this upper bound condition to reduce computational burden.

Condition \hyperref[cond_A3]{\rm (A3)} gives a lower bound for nonzero signals, which is called the {\it beta-min} condition.
In general, this type of condition is necessary to not miss any nonzero signals. 
Together with Condition \hyperref[cond_A2]{\rm (A2)}, it allows the minimum signal to go to zero as $n\to\infty$.
We note here that we have $\log r$ term in the numerator of \eqref{beta-min}, which would be $\log p$ if the group structure is unknown.
Thus, it says that incorporating the group structure into the model indeed benefits  variable selection.

Condition \hyperref[cond_A4]{\rm (A4)} suggests appropriate conditions for hyperparameters.
In  high-dimensional settings $r \ge n$, it implies that the variance of the slab part, $\tau^2$, grows to infinity, which is a common assumption in the literature \citep{Naveen:2014,Naveen:2018,martin2017empirical}.
It also assumes that the prior inclusion probability for $G_j$, $q_j$, has the same rate with $r^{-1}$ up to some constant.
In general, when the group structure is unknown, it has been assumed that the prior inclusion probability for each variable is proportional to $p^{-c}$ for some constant $c\ge 1$ \citep{narisetty2014bayesian,castillo2015bayesian}.
In our case, the number of variables, $p$, is replaced with the number of groups, $r$.

\begin{theorem}[No super set]\label{thm:nosuper}
	Under Conditions \hyperref[cond_A1]{\rm (A1)}, \hyperref[cond_A2]{\rm (A2)} and \hyperref[cond_A4]{\rm (A4)}, 
	\bea
	\pi  \big( Z \supsetneq t \mid E  ,\, |G_Z| \le m_n \big) &\overset{P}{\lra}& 0  , \,\, \text{ as } n\to\infty  .
	\eea
\end{theorem}
In Theorem \ref{thm:nosuper} and following theorems, we consider the conditional posterior given the event of reasonably large models, $\{ Z: |G_Z| \le m_n\}$.
As stated before, this can be easily satisfied by restricting the prior \eqref{model2} on $\{ Z: |G_Z| \le m_n \}$.
Similar assumptions restricting the model size have been used in  \cite{liang2013bayesian,Naveen:2018} and \cite{LLL:2019}, to name a few.

Theorem \ref{thm:nosuper} says that, asymptotically, our posterior will not overfit the model, i.e., not include unnecessarily many variables.
Of course, it does not guarantee that the posterior will concentrate on the true model.
To capture every significant variable, we require the magnitudes of nonzero entries in $\beta_0$ not to be too small.
In general, if a nonzero entry in $\beta_0$ is too close to zero, it cannot be selected by any method.
Condition \hyperref[cond_A3]{\rm (A3)} gives a lower bound for the magnitudes to be detected.
Theorem \ref{thm:ratio} shows that with an appropriate lower bound condition, the true model $t$ will be the mode of the posterior.

\begin{theorem}[Posterior ratio consistency]\label{thm:ratio}
	Under Conditions \hyperref[cond_A1]{\rm (A1)}--\hyperref[cond_A4]{\rm (A4)} with $c_0 = \{(1-\epsilon_0)\lambda\}^{-1} \big\{  7 + 5 \{(1-\epsilon_0)\lambda\}^{-1}  \big\}$ for some small constant $\epsilon_0>0$,
	\bea
	\max_{k \neq t} \frac{\pi  \big( Z = k \mid E  ,\, |G_Z| \le m_n \big)}{\pi  \big( Z =t \mid E  ,\, |G_Z| \le m_n \big)} &\overset{P}{\lra}& 0  , \,\, \text{ as } n\to\infty .
	\eea
\end{theorem}

Posterior ratio consistency is a useful property especially when we are interested in the point estimation with the posterior mode.
However, it does not give a suitable answer to how large probability the posterior puts on the true model.
In the following theorem, we state that our posterior achieves {\it strong selection consistency}, which resolves the question.
By strong selection consistency, we mean that the posterior probability assigned to the true model $t$ grows to 1 as we observe more data.
One can see that strong selection consistency implies posterior ratio consistency, but not vice versa.
Theorem \ref{thm:selection} says that, to achieve strong selection consistency, a slightly larger lower bound for the magnitudes of nonzero entries in $\beta_0$ is required compared to that in Theorem \ref{thm:ratio}.
\begin{theorem}[Strong selection consistency]\label{thm:selection}
	Under Conditions \hyperref[cond_A1]{\rm (A1)}--\hyperref[cond_A4]{\rm (A4)} with $c_0 = \{(1-\epsilon_0)\lambda\}^{-1}  \big\{  17 + 5\{(1-\epsilon_0)\lambda\}^{-1}  \big\}$ for some small constant $\epsilon_0>0$,
	\bea
	\pi  \big( Z =t \mid E  ,\, |G_Z| \le m_n \big) &\overset{P}{\lra}& 1  , \,\, \text{ as } n\to\infty  .
	\eea
\end{theorem}

\section{Posterior Computation for Group Selection}\label{sec:post_comp}
The logistic model \eqref{logistic_model} is equivalent to letting 
\begin{equation} \label{indicator_func}
E_i = \mathbbm{1}(Y_i \ge 0),
\end{equation}
where $\mathbbm 1(\cdot)$ is the indicator function and $Y_i$ is an underlying continuous variable that has a univariate logistic density with
\begin{equation} \label{logistic_density}
\pi(Y_i \mid  x_i,\beta_{G_1},\ldots,\beta_{G_r}) = \frac{\exp\left\{-\left(Y_i - \sum_{j=1}^r x_{iG_j}^T\beta_j\right)\right\}}{\left[1+\exp\left\{-\left(Y_i - \sum_{j=1}^r x_{iG_j}^T\beta_{G_j}\right)\right\}\right]^2}.
\end{equation} 
As noted in \cite{Albert:Chib:1993,Brien:Dunson:2004} and \cite{Naveen:2018}, an excellent approximation to the univariate logistic density \eqref{logistic_density} can be obtained using the following $t$-distribution with $\nu$ degrees of freedom and scale parameter $\sigma_0^2$,
\begin{equation*}
\pi_\nu(Y_i \mid x_i, \beta_{G_1},\ldots,\beta_{G_r}) = \frac{\Gamma((\nu+1)/2)}{\Gamma(\nu/2)\sqrt{\nu\pi}\sigma_0}\left\{1 + \frac{\left(Y_i - \sum_{j=1}^r x_{iG_j}^T\beta_{G_j}\right)^2}{\nu\sigma_0^2}\right\}^{-(\nu+1)/2}.
\end{equation*}
Following \cite{Albert:Chib:1993,Brien:Dunson:2004} and \cite{Naveen:2018}, by setting $\sigma_0^2 = \pi^2(\nu - 2)/3\nu$ and $\nu = 7.3$, the resulting $t$-distribution is nearly indistinguishable from the logistic distribution in \eqref{logistic_density}. As we shall demonstrate subsequently, one can exploit this approximation to develop efficient algorithms for posterior computation instead of directly sampling from the logistic distribution. In particular, the $t$-distribution can be equivalently represented as a mixture of normal distributions as follows, for $1 \le i \le n$,
\begin{equation} \label{scale_normal}
Y_i \sim N(\sum_{j=1}^r x_{iG_j}^T\beta_{G_j}, s_i^2), \quad 1/s_i^2 \sim \mbox{Gamma}\left(\nu/2, \sigma_0^2\nu/2\right).
\end{equation}
Let $W= \mbox{diag}\left\{s_1^{-2}, s_2^{-2}, \ldots, s_n^{-2}\right\}$. 
By \eqref{basad}, \eqref{indicator_func} and \eqref{scale_normal}, we have the joint posterior for $(\beta, Y, Z, W)$ given by 
\begin{align} \label{joint_posterior}
\pi(\beta,  Y, Z, W \mid E) \propto& \prod_{i = 1}^{n} N \big( Y_i \mid \sum_{j=1}^r x_{iG_j}^T\beta_{G_j}, s_i^2 \big) \mathbbm 1 \big\{ E_i = \mathbbm{1}\left(Y_i \ge 0\right) \big\} \pi(s_i^2) \nonumber\\
& \times \prod_{j = 1}^r \pi(Z_j)\left\{ (1 - Z_j) \delta_0\left(\beta_{G_j}\right) + Z_jN_{|G_j|}(0, \tau^2I_{|G_j|}) \right\}.
\end{align}
It follows that
\begin{align}
\pi(\beta,  Y, Z, W \mid E)
\propto& \exp\left\{-(Y -X\beta)^TW(Y - X\beta)/2\right\}\prod_{i = 1}^n \mathbbm 1 \big\{  E_i = \mathbbm{1}\left(Y_i \ge 0\right) \big\} \nonumber\\
&\times \prod_{i = 1}^n (s_i^2)^{-1/2 - \nu/2 - 1} \exp \left\{-\sigma_0^2\nu/ (2s_i^2)\right\} \nonumber\\
&\times \prod_{j: Z_j=1}(2\pi\tau^2)^{-|G_j|/2}\exp\left\{-\beta_{G_j}^T\beta_{G_j}/(2\tau^2)\right\}\prod_{j: Z_j=0}\mathbbm 1 \left(\beta_{G_j} = 0\right) \nonumber\\
&\times \left\{ q_j/(1-q_j) \right\}^{\sum_{j=1}^r Z_j}.
\end{align}

\subsection{Gibbs Sampler}
For posterior computation, we first consider the standard approach to draw samplers from $\pi(\beta,  Y, Z, W|E)$ using the following posterior distributions. For $1 \le j \le r$, let $\beta_{-G_j}$ denote the $\beta$ vector without the $j$th group. Let $X_{-G_j}$ denote the covariate matrix corresponding to $\beta_{-G_j}$. 
\begin{itemize}
	\item For  $1 \le i \le n$, generate $Y_i$ via the following conditional distribution,
	\begin{equation} \label{gibbs_Y}
	\pi(Y_i \mid \mbox{rest}) \propto  \begin{cases}
	N(Y_i \mid x_i^T \beta, s_i^2) \mathbbm{1}\left(Y_i > 0\right), \quad\mbox{if } E_i = 1,\\
	N(Y_i \mid x_i^T \beta, s_i^2) \mathbbm{1}\left(Y_i < 0\right), \quad\mbox{if } E_i = 0 .
	\end{cases} 
	\end{equation}
	\item For  $1 \le i \le n$, generate $s_i $ as follows,
	\begin{equation} \label{gibbs_s}
	s_i^2 \mid \mbox{rest} \sim \mbox{Inverse-Gamma} \left(1/2 + \nu/2, (Y_i - x_i^T\beta)^2/2 + \sigma_0^2\nu/2\right) 
	\end{equation}
	\item For $1 \le j \le r$, let $\Sigma_{G_j} = \left(X_{G_j}^T W X_{G_j} + I_{|G_j|}/\tau^2\right)^{-1}$ and $\mu_{G_j} = \Sigma_{G_j} X_{G_j}^TW (Y - X_{-G_j}\beta_{-G_j})$. 
	The conditional distribution of $Z_j$ is given by
	\begin{align} \label{gibbs_gamma}
	Z_j \mid Y, W, Z_{-j}, \beta_{-G_j} \sim \mbox{Bernoulli}\left(1 - \left(1 + \frac {q_j} {1-q_j}(\tau^2)^{-|G_j|/2}|\Sigma_{G_j}|^{1/2}\exp\left\{1/2\mu_{G_j}^T\Sigma_{G_j}^{-1}\mu_{G_j}\right\}\right)^{-1}\right).
	\end{align}	
	\item For $1 \le j \le r$, generate $\beta_{G_j}$ based on the following spike and slab distribution,
	\begin{align} \label{gibbs_beta}
	\beta_{G_j} \mid \mbox{rest} \sim (1 - Z_j)\delta_0  + Z_jN_{|G_j|}(\mu_{G_j}, \Sigma_{G_j}).
	\end{align}
\end{itemize}
Note that when sampling $Z_j$, we are using the conditional posterior after integrating out $\beta_{G_j}$ rather than using the full conditional posterior.
The same trick was used in \cite{Yang:Naveen:2018} and \cite{Xu:Ghosh:2015} to ensure that the Markov chain will be irreducible and converge. 
Otherwise, reversible-jump Markov chain Monte Carlo \citep{Green:1995} will be needed, which might be impractical under high-dimensional settings. 
To estimate the highest posterior probability model, we record the model selected at each simulation and tabulate them to find the model that appears most often.

\subsection{Neuronized Reparameterization of Spike and Slab Prior}

In this section, we suggest an alternative MCMC algorithm using the neuronized prior \citep{Shin:Liu:2018}.
\citet{Shin:Liu:2018} proposed neuronized priors that provide a unified form of shrinkage priors including continuous shrinkage priors, continuous and discrete spike and slab priors. 
Compared with classic spike and slab priors, the neuronized priors achieve the same explicit variable selection without employing any latent indicator variable, which results in more efficient MCMC algorithms. 
In the form of neuronized priors, each regression coefficient is reparameterized as a product of a weight parameter and a transformed scale parameter via an activation function. 
Adapted to our group selection setting, we define the neuronized prior as follows, for $1 \le j \le r$,
\begin{equation} \label{neuronized}
\beta_{G_j}(\alpha_j, w_{|G_j|}) := \mathbbm{1}\left(\alpha_j - \alpha_0 \ge 0\right) w_{G_j},
\end{equation} 
where the scale parameter $\alpha_j$ follows the standard normal distribution and the weight vector $w_{G_j}$ follows $N_{|G_j|}\left(0, \tau^2\right)$. Note that in accordance with our model with latent variable, $1 - q = P(Z_j = 0) = P(\beta_{G_j} = 0) = P(\alpha_j < \alpha_0) = \Phi(\alpha_0)$. Therefore, we set $\alpha_0 = \Phi^{-1}(1-q)$. 

Under the neuronized reparameterization, the posterior distribution of $\alpha$ and $w$ is given by
\begin{align}
\pi(\alpha, w \mid \mbox{rest}) \propto \exp\left\{-(Y-X\beta(\alpha,w))^TW(Y-X\beta(\alpha,w))/2 - \alpha^T\alpha/2 - w^Tw/2\tau^2\right\},
\end{align}
where $\alpha = \left(\alpha_1, \ldots, \alpha_r\right)^T$, $w = \left(w_{G_1}^T, \ldots, w_{G_r}^T\right)^T$ and $\beta(\alpha,w) = D_\alpha w$ with $$D_\alpha = \mbox{diag}\{\underbrace{\mathbbm{1}\left(\alpha_1 - \alpha_0 \ge 0\right),\ldots, \mathbbm{1}\left(\alpha_1 - \alpha_0 \ge 0\right)}_{|G_1|}, \ldots, \underbrace{\mathbbm{1}\left(\alpha_r - \alpha_0 \ge 0\right),\ldots, \mathbbm{1}\left(\alpha_r - \alpha_0 \ge 0\right)}_{|G_r|}\}.$$
Therefore, the conditional posterior distribution of $w$ given the rest is multivariate normal expressed as
\begin{equation} \label{gibbs_w}
w \mid \mbox{rest} \sim N((D_\alpha X^TWXD_\alpha + I/\tau^2)^{-1}D_\alpha X^TWY, (D_\alpha X^TWXD_\alpha + I/\tau^2)^{-1}).
\end{equation}
As illustrated in \cite{Shin:Liu:2018}, if directly sampling $w$ from its conditional posterior distributions in high dimensions, the numerical calculation will involve computing the inverse matrix which is very expensive. Hence, we adopt the fast sampling procedure \citep{Bhattacharya:2016} to reduce the computational complexity. Before presenting the MCMC algorithm for the neuronized version, we show below that the conditional distribution $\pi(\alpha_j|\mbox{rest})$ is a mixture of two truncated normals, which can be sampled directly. Note that by \eqref{neuronized}, for $1 \le j \le r$, the conditional posterior of $\alpha_j$ given the rest is
\begin{align*}
&\pi(\alpha_j \mid \mbox{rest}) \\
\propto& (2\pi)^{-1/2}\exp\left\{- (r_j - X_{G_j}\mathbbm{1}\left(\alpha_j - \alpha_0 \ge 0\right)w_{G_j})^TW(r_j - X_{G_j}\mathbbm{1}\left(\alpha_j - \alpha_0 \ge 0\right)w_{G_j})/2- \alpha_j^2/2\right\},
\end{align*}
where $r_{j} = Y - X\beta(\alpha, w) + X_{G_j}\mathbbm{1}\left(\alpha_j - \alpha_0 \ge 0\right)w_{G_j}$. It follows that
\begin{equation}
\pi(\alpha_j \mid \mbox{rest}) \propto  \begin{cases}
(2\pi)^{-1/2}\exp\left\{- (r_j - X_{G_j}w_{G_j})^TW(r_j - X_{G_j}w_{G_j})/2- \alpha_j^2/2\right\}, \quad\mbox{if } \alpha_j \ge \alpha_0,\\
(2\pi)^{-1/2}\exp\left\{- r_j^TWr_j/2- \alpha_j^2/2\right\}, \quad\mbox{if } \alpha_j < \alpha_0.
\end{cases} 
\end{equation}
Therefore, 
\begin{equation} \label{gibbs_alpha}
\alpha_j \mid \mbox{rest} \sim \kappa N_{tr}(0, 1; -\infty, \alpha_0) + (1-\kappa)N_{tr}(0, 1; \alpha_0, \infty),
\end{equation}
where $$\kappa = \frac{\Phi(\alpha_0)\exp\left\{- r_j^TWr_j/2\right\}}{\Phi(\alpha_0)\exp\left\{- r_j^TWr_j/2\right\} + (1-\Phi(\alpha_0))\exp\left\{- (r_j - X_{G_j}w_{G_j})^TW(r_j - X_{G_j}w_{G_j})/2\right\}}.$$
Now we are ready to present the following MCMC algorithm under the neuronized reparameterization. For each iteration,
\begin{itemize}
	\item Sample $w$ from \eqref{gibbs_w}.
	\item Update $r_e = Y - X\beta(\alpha, w).$
	\item[] For $j = 1, 2, \ldots, r$
	\begin{itemize}
		\item Set $r_j = r_e + X_{G_j}\mathbbm{1}\left(\alpha_j - \alpha_0 \ge 0\right)w_{G_j}$.
		\item Sample $\alpha_j$ from \eqref{gibbs_alpha}.
		\item Update $r_e = r_j - X_{G_j}\mathbbm{1}\left(\alpha_j - \alpha_0 \ge 0\right)w_{G_j}$.
	\end{itemize}
	\item[] For $i = 1, 2, \ldots, n$
	\begin{itemize}
		\item Sample $Y_i$ and $s_i$ from \eqref{gibbs_Y} and \eqref{gibbs_s} respectively.
	\end{itemize}
\end{itemize}

\section{Simulation Studies}\label{sec:sim}
In this section, we demonstrate the performance of the proposed method in various settings. We closely follow but slightly modify the simulation settings in \cite{Breheny:2015, Naveen:2018} and \cite{Yang:Naveen:2018}. Let $X$ denote the design matrix whose first $|t|$ groups of columns correspond to the active covariates for which we have nonzero coefficients, while the rest correspond to the inactive ones with zero coefficients. 
In all the simulation settings, we generate $x_i \overset{i.i.d.}{\sim} N_p(0, \Sigma )$ for $i=1,\ldots,n$ under the following three different cases of $\Sigma$.
\begin{enumerate} 
	\item Case 1: Isotropic design, where $\Sigma = I_p$. There is no correlation either within or between different groups. 
	\item Case 2: Compound symmetry design, where $\Sigma_{ij} = 0.5$, if $i \neq j$ and $\Sigma_{ii} = 1,$ for all $1 \le i \le j \le p.$ The correlations at both the group and within group levels are equally 0.5.
	\item Case 3: Autoregressive correlated design; where $\Sigma_{ij} = 0.5^{|i-j|},$ for all $1 \le i \le j \le p$. The correlations at both the group and within group levels are set to different values.
\end{enumerate}
For given $X$ and $1\le i \le n$, we sample $E_i $ from the following logistic model as in \eqref{logistic_model}
\bea
P\left(E_i = 1 \mid x_i\right) = \frac{\exp\left\{\sum_{j=1}^r x_{iG_j}^T\beta_{G_j}\right\}}{1+\exp\left\{\sum_{j=1}^r x_{iG_j}^T\beta_{G_j}\right\}}  .
\eea
Following the simulation settings in \cite{Breheny:2015} and \cite{Yang:Naveen:2018}, we consider the following three designs, each with the same sample size $n = 100$ and number of groups being either $r=50$ or $r=100$. The group size is taken to be one of 4, 5 or 6 randomly with equal probability.
\begin{enumerate}
	\item Design 1 (Baseline design): The number of groups $r = 50$ and $|t| = 3.$
	\item Design 2 (Dense model design):  The number of groups $r = 50$ and $|t| = 6.$
	\item Design 3 (High dimensional design): The number of groups $r = 100$ and $|t| = 3.$
\end{enumerate}
We investigate the following four settings for the true coefficient vector $\beta_0$ to include different combinations of small and large signals.
\begin{enumerate}
	\item Setting 1: All the entries of $\beta_{0, G_t}$ are generated from $\mbox{Unif}(0.5,1.5)$.
	\item Setting 2: All the entries of $\beta_{0, G_t}$ are set to 1.5.
	\item Setting 3: All the entries of $\beta_{0, G_t}$ are generated from $\mbox{Unif}(1.5, 3)$.
	\item Setting 4: All the entries of $\beta_{0, G_t}$ are set to 3.
\end{enumerate}
We will refer to our proposed method as Group Spike and Slab (GSS) and the estimates computed from standard Gibbs sampling algorithms and neuronized version as GSS-G and GSS-N, respectively. The performance of GSS will then be compared with other existing methods including group lasso (grLasso), group SCAD (grSCAD), group MCP (grMCP), and group exponential Lasso (gel). 
The tuning parameters in the regularization approaches are chosen by 10-fold cross-validation.
For GSS, the hyperparameters are set at $\tau^2 = \max\{1, 0.01n^{-1}r^{2+2\times0.01}\}$ and $q = 1/r$.
The initial state for $\gamma$ was set by randomly taking three groups to be active and the remaining to be inactive.
For posterior inference, $2,000$ posterior samples were drawn with a burn-in period of $2,000$.
The indices having posterior inclusion probability larger than $0.5$ were included in the final model.
The resulting model is called the median probability model, and when there is a model with posterior probability larger than $1/2$, it coincides with the posterior mode \citep{barbieri2004optimal}.

To evaluate the performance of variable selection, the sensitivity, specificity, Matthews correlation coefficient (MCC), the number of errors (\#Error) and mean-squared prediction error (MSPE) are reported at Tables \ref{table:comp1} to \ref{table:comp9}, where each simulation setting is repeated for 50 times. 
The criteria are defined as
\bea
\text{Sensitivitiy}  &=&     \frac{TP}{TP+FN} ,   \\
\text{Specificity}  &=&    \frac{TN}{TN+FP}   ,  \\
\text{MCC}  &=&     \frac{TP \times TN - FP\times FN}{\sqrt{(TP+FP)(TP+FN)(TN+FP)(TN+FN)}}  ,	  \\
\text{MSPE}  &=&    \frac{1}{n_{\rm test}} \sum_{i=1}^{n_{\rm test}}     \big( \hat{Y}_i - Y_{{\rm test},i}  \big)^2  ,
\eea
where TP, TN, FP and FN are true positive, true negative, false positive and false negative, respectively.
Here we denote $\hat{Y}_i = x_i^T \hat{\beta}$, where $\hat{\beta}$ is the estimated coefficient based on each method.
For Bayesian methods, the usual GLM estimates based on the selected support are used as $\hat{\beta}$.
We generated test samples $Y_{{\rm test}, 1},\ldots, Y_{{\rm test}, n_{\rm test}}$ with $n_{\rm test}=100$ to calculate the MSPE.

\begin{table}[!tb]
	\centering\footnotesize
	\caption{
		The summary statistics for Design 1 (Baseline design) are represented for each setting of the true regression coefficients under the first isotropic case ($\Sigma = I_p$).
		Different setting means different choice of the true coefficient $\beta_0$.}
	\vspace{.15cm}
	\begin{tabular}{c c c c c |   c c c c}
		\hline 
		& \multicolumn{4}{c}{ Setting 1 } & \multicolumn{4}{c}{ Setting 2 } \\ 
		& Sensitivity & Specificity & MCC & MSPE & Sensitivity & Specificity & MCC & MSPE \\ \hline
		GGS-N &1 &1 &1 &0.16 &1 &1 &1 &0.13   \\ 
		GGS-G  &0.97 &1 &0.98 &0.17&1 &1 &1 &0.12 \\ 
		grLasso  &1 &0.53 &0.25 &0.18 &1 &0.68 &0.34 &0.16  \\ 
		grSCAD &1 &0.62 &0.30 &0.17 &1 &0.70 &0.35 &0.15 \\ 
		grMCP &1 &0.53 &0.25 &0.25 &1 &0.68 &0.34 &0.16  \\ 
		gel &1 &0.68 &0.34 &0.31 &1 &0.72 &0.37 &0.18 \\\hline \hline 
		& \multicolumn{4}{c}{ Setting 3 } & \multicolumn{4}{c}{ Setting 4 } \\ 
		& Sensitivity & Specificity & MCC & MSPE & Sensitivity & Specificity & MCC & MSPE \\ \hline
		GGS-N &1 &1 &1 &0.12 &1 &1 &1 &0.10   \\ 
		GGS-G  &1 &1 &1 &0.12  &1 &1 &1 &0.10   \\ 
		grLasso  &1 &0.70 &0.35 &0.14 &1 &0.79 &0.43 &0.13  \\ 
		grSCAD &1 &0.79 &0.43 &0.14 &1 &0.79 &0.42 &0.14  \\ 
		grMCP &1 &0.79 &0.42 &0.15 &1 &0.76 &0.40 &0.13  \\ 
		gel &1 &0.74 &0.39 &0.18 &1 &0.77 &0.41 &0.12 \\\hline
	\end{tabular}\label{table:comp1}
\end{table}

\begin{table}[!tb]
	\centering\footnotesize
	\caption{
		The summary statistics for Design 1 (Baseline design) are represented for each setting of the true regression coefficients under the second compound symmetry covariance case. Different setting means different choice of the true coefficient $\beta_0$.}
	\vspace{.15cm}
	\begin{tabular}{c c c c c |   c c c c}
		\hline 
		& \multicolumn{4}{c}{ Setting 1 } & \multicolumn{4}{c}{ Setting 2 } \\ 
		& Sensitivity & Specificity & MCC & MSPE & Sensitivity & Specificity & MCC & MSPE \\ \hline
		GGS-G  &0.33 &1 &0.57 &0.11 &0.53 &0.98 &0.59 &0.10 \\ 
		GGS-N &0.54 &0.98 &0.57 &0.10 &0.58 &0.96 &0.60 &0.09   \\ 
		grLasso  &0.32 &0.74 &0.04 &0.11 &1 &0.62 &0.32 &0.10  \\ 
		grSCAD &0.33 &0.91 &0.21 &0.12 &1 &0.62 &0.30 &0.11  \\ 
		grMCP &0.33 &0.83 &0.10 &0.12 &1 &0.79 &0.43 &0.12  \\ 
		gel &0.67 &0.94 &0.48 &0.12 &0.33 &0.85 &0.12 &0.11 \\\hline \hline 
		& \multicolumn{4}{c}{ Setting 3 } & \multicolumn{4}{c}{ Setting 4 } \\ 
		& Sensitivity & Specificity & MCC & MSPE & Sensitivity & Specificity & MCC & MSPE \\ \hline
		GGS-G &0.57 &0.99 &0.65 &0.08 &0.57 &0.99 &0.65 &0.09   \\ 
		GGS-N &0.59 &0.98 &0.65 &0.10 &0.60 &0.98 &0.71 &0.06   \\ 
		grLasso  &1 &0.53 &0.25 &0.10 &1 &0.53 &0.25 &0.09  \\ 
		grSCAD &1 &0.51 &0.20 &0.11 &1 &0.53 &0.25 &0.10 \\ 
		grMCP &1  &0.68 &0.29 &0.10 &1 &0.70 &0.35 &0.10 \\ 
		gel &0.33 &0.86 &0.18 &0.11 &0.33 &0.91 &0.20 &0.09 \\\hline 
	\end{tabular}\label{table:comp2}
\end{table}

\begin{table}[!tb]
	\centering\footnotesize
	\caption{
		The summary statistics for Design 1 (Baseline design) are represented for each setting of the true regression coefficients under the third autoregressive covariance case.
		Different setting means different choice of the true coefficient $\beta_0$.}
	\vspace{.15cm}
	\begin{tabular}{c c c c c |   c c c c}
		\hline 
		& \multicolumn{4}{c}{ Setting 1 } & \multicolumn{4}{c}{ Setting 2 } \\ 
		& Sensitivity & Specificity & MCC & MSPE & Sensitivity & Specificity & MCC & MSPE \\ \hline
		GGS-G &1 &1 &1 &0.12 &1 &1 &1 &0.12   \\ 
		GGS-N  &0.93 &1 &0.96 &0.12 &0.93 &1 &0.96 &0.12  \\ 
		grLasso  &1 &0.61 &0.30 &0.15 &1 &0.61 &0.30 &0.14  \\ 
		grSCAD &1 &0.60 &0.29 &0.15 &1 &0.59 &0.29 &0.14  \\ 
		grMCP &1 &0.70 &0.36 &0.14 &1 &0.70 &0.35 &0.12  \\ 
		gel &1 &0.74 &0.39 &0.14 &1 &0.74 &0.38 &0.13  \\\hline \hline 
		& \multicolumn{4}{c}{ Setting 3 } & \multicolumn{4}{c}{ Setting 4 } \\ 
		& Sensitivity & Specificity & MCC & MSPE & Sensitivity & Specificity & MCC & MSPE \\ \hline
		GGS-G &1 &1 &1 &0.10 &1 &1 &1 &0.06   \\ 
		GGS-N &1 &1 &1 &0.10 &1 &1 &1 &0.06   \\ 
		grLasso  &1 &0.62 &0.30 &0.14 &1 &0.62 &0.30 &0.14  \\ 
		grSCAD &1 &0.68 &0.34 &0.15 &1 &0.71 &0.34 &0.15  \\ 
		grMCP &1 &0.70 &0.35 &0.14 &1 &0.70 &0.35 &0.12  \\ 
		gel &1 &0.74 &0.49 &0.14 &1 &0.72 &0.37 &0.13  \\\hline 
	\end{tabular}\label{table:comp3}
\end{table}

\begin{table}[!tb]
	\centering\footnotesize
	\caption{
		The summary statistics for Design 2 (Dense model design) are represented for each setting of the true regression coefficients under the first isotropic case ($\Sigma = I_p$).
		Different setting means different choice of the true coefficient $\beta_0$.}
	\vspace{.15cm}
	\begin{tabular}{c c c c c |   c c c c}
		\hline 
		& \multicolumn{4}{c}{ Setting 1 } & \multicolumn{4}{c}{ Setting 2 } \\ 
		& Sensitivity & Specificity & MCC & MSPE & Sensitivity & Specificity & MCC & MSPE \\ \hline
		GGS-G &0.46 &1 &0.61 &0.23 &0.55 &1 &0.71 &0.22   \\ 
		GGS-N  &0.44 &1 &0.61 &0.23 &0.51 &1 &0.70 &0.23  \\ 
		grLasso  &1 &0.48 &0.31 &0.17 &1 &0.39 &0.27 &0.19  \\ 
		grSCAD &1 &0.47 &0.31 &0.18 &1 &0.28 &0.31 &0.19  \\ 
		grMCP &1 &0.57 &0.37 &0.15 &1 &0.26& 0.25 &0.17  \\ 
		gel &1 &0.57 &0.37 &0.17 &1 &0.55 &0.35 &0.20 \\\hline \hline 
		& \multicolumn{4}{c}{ Setting 3 } & \multicolumn{4}{c}{ Setting 4 } \\ 
		& Sensitivity & Specificity & MCC & MSPE & Sensitivity & Specificity & MCC & MSPE \\ \hline
		GGS-G &0.65 &1 &0.76 &0.21 &0.72 &1 &0.83 &0.19   \\ 
		GGS-N &0.59 &1 &0.74 &0.32 &0.69 &1 &0.82 &0.21   \\ 
		grLasso  &1 &0.39 &0.27 &0.18 &1 &0.39 &0.27 &0.19  \\ 
		grSCAD &1 &0.48 &0.31 &0.18 &1 &0.57 &0.39 &0.20  \\ 
		grMCP &1 &0.45 &0.30 &0.15 &1 &0.48 &0.31 &0.20  \\ 
		gel &1 &0.55 &0.35 &0.19 &0.83 &0.57 &0.26 &0.24 \\\hline 
	\end{tabular}\label{table:comp4}
\end{table}

\begin{table}[!tb]
	\centering\footnotesize
	\caption{
		The summary statistics for Design 2 (Dense model design) are represented for each setting of the true regression coefficients under the second compound symmetry covariance case.
		Different setting means different choice of the true coefficient $\beta_0$.}
	\vspace{.15cm}
	\begin{tabular}{c c c c c |   c c c c}
		\hline 
		& \multicolumn{4}{c}{ Setting 1 } & \multicolumn{4}{c}{ Setting 2 } \\ 
		& Sensitivity & Specificity & MCC & MSPE & Sensitivity & Specificity & MCC & MSPE \\ \hline
		GGS-G &0.29 &0.98 &0.33 &0.17 &0.31 &0.98 &0.33 &0.15   \\ 
		GGS-N  &0.30 &0.99 &0.35 &0.14 &0.35 &0.99 &0.38 &0.14 \\ 
		grLasso  &0.83 &0.55 &0.25 &0.11 &0.83 &0.56 &0.25 &0.11  \\ 
		grSCAD &0.83 &0.64 &0.31 &0.12 &0.82 &0.64 &0.31 &0.12  \\ 
		grMCP &0.17 &0.82 &0.01 &0.13 &0.17 &0.82 &0.01 &0.14  \\ 
		gel &0.16 &0.91 &0.08 &0.13 &0.17 &0.92 &0.08 &0.13 \\\hline \hline 
		& \multicolumn{4}{c}{ Setting 3 } & \multicolumn{4}{c}{ Setting 4 } \\ 
		& Sensitivity & Specificity & MCC & MSPE & Sensitivity & Specificity & MCC & MSPE \\ \hline
		GGS-G &0.25 &0.97 &0.27 &0.18 &0.26 &0.97 &0.30 &0.15   \\ 
		GGS-N &0.25 &0.98 &0.27 &0.14 &0.31 &0.98 &0.35 &0.14   \\ 
		grLasso  &0.82 &0.55 &0.25 &0.11 &0.83 &0.54 &0.25 &0.11  \\ 
		grSCAD &0.83 &0.64 &0.31 &0.12 &0.83 &0.63 &0.30 &0.12  \\ 
		grMCP &0.17 &0.82 &0.01 &0.13 &0.17 &0.82 &0.01 &0.13  \\ 
		gel &0.17 &0.90 &0.08 &0.13 &0.17 &0.91 &0.08 &0.13  \\\hline 
	\end{tabular}\label{table:comp5}
\end{table}

\begin{table}[!tb]
	\centering\footnotesize
	\caption{
		The summary statistics for Design 2 (Dense model design) are represented for each setting of the true regression coefficients under the third autoregressive covariance case.
		Different setting means different choice of the true coefficient $\beta_0$.}
	\vspace{.15cm}
	\begin{tabular}{c c c c c |   c c c c}
		\hline 
		& \multicolumn{4}{c}{ Setting 1 } & \multicolumn{4}{c}{ Setting 2 } \\ 
		& Sensitivity & Specificity & MCC & MSPE & Sensitivity & Specificity & MCC & MSPE \\ \hline
		GGS-G &0.65 &1 &0.79 &0.22 &0.73 &1 &0.84 &0.14   \\ 
		GGS-N  &0.59 &0.99 &0.72 &0.24 &0.71 &1 &0.81 &0.14 \\ 
		grLasso  &1 &0.57 &0.37 &0.17 &1 &0.45 &0.30 &0.15  \\ 
		grSCAD &1 &0.57 &0.37 &0.17 &1 &0.57 &0.37 &0.16  \\ 
		grMCP &0.83 &0.64 &0.31 &0.21 &1 &0.64 &0.42 &0.15  \\ 
		gel &0.83 &0.68 &0.34 &0.16 &0.83 &0.68 &0.34 &0.15 \\\hline \hline 
		& \multicolumn{4}{c}{ Setting 3 } & \multicolumn{4}{c}{ Setting 4 } \\ 
		& Sensitivity & Specificity & MCC & MSPE & Sensitivity & Specificity & MCC & MSPE \\ \hline
		GGS-G &0.77 &1 &0.86 &0.15 &0.77 &1 &0.85 &0.13   \\ 
		GGS-N &0.72 &1 &0.84 &0.15 &0.73 &1 &0.81 &0.14   \\ 
		grLasso  &1 &0.57 &0.37 &0.14&1 &0.45 &0.30 &0.14  \\ 
		grSCAD &1&0.57 &0.37 &0.15&1 &0.57 &0.37 &0.15 \\ 
		grMCP &1&0.66 &0.43 &0.16&0.83 &0.70 &0.36 &0.16 \\ 
		gel &0.67&0.59&0.17&0.14&0.83 &0.59 &0.28 &0.15 \\\hline 
	\end{tabular}\label{table:comp6}
\end{table}

\begin{table}[!tb]
	\centering\footnotesize
	\caption{
		The summary statistics for Design 3 (High dimensional design) are represented for each setting of the true regression coefficients under the first isotropic case ($\Sigma = I_p$).
		Different setting means different choice of the true coefficient $\beta_0$.}
	\vspace{.15cm}
	\begin{tabular}{c c c c c |   c c c c}
		\hline 
		& \multicolumn{4}{c}{ Setting 1 } & \multicolumn{4}{c}{ Setting 2 } \\ 
		& Sensitivity & Specificity & MCC & MSPE & Sensitivity & Specificity & MCC & MSPE \\ \hline
		GGS-G &0.67 &1 &0.81 &0.13&1 &1 &1 &0.07   \\ 
		GGS-N  &0.63 &1 &0.80 &0.14 &0.73 &1 &0.84 &0.11 \\ 
		grLasso  &1 &0.65 &0.23 &0.20&1 &0.70 &0.26 &0.19  \\ 
		grSCAD &1 &0.73 &0.28 &0.18&1 &0.81 &0.34 &0.17  \\ 
		grMCP &1 &0.81 &0.34 &0.24 &1 &0.81 &0.34 &0.21  \\ 
		gel &1 &0.86 &0.39 &0.31 &1 &0.89 &0.44 &0.29 \\\hline \hline 
		& \multicolumn{4}{c}{ Setting 3 } & \multicolumn{4}{c}{ Setting 4 } \\ 
		& Sensitivity & Specificity & MCC & MSPE & Sensitivity & Specificity & MCC & MSPE \\ \hline
		GGS-G &1 &1 &1 &0.05&1 &1 &1 &0.02   \\ 
		GGS-N &1 &1&1 &0.05 &1 &1 &1 &0.02   \\ 
		grLasso  &1 &0.70 &0.26 &0.12 &1 &0.70 &0.26 &0.13  \\ 
		grSCAD &1 &0.86 &0.39 &0.12 &1 &0.81 &0.34 &0.12 \\ 
		grMCP &1 &0.77 &0.30 &0.12 &1 &0.77 &0.30 &0.12  \\ 
		gel &1 &0.82 &0.34 &0.17 &1 &0.81 &0.34 &0.17 \\\hline 
	\end{tabular}\label{table:comp7}
\end{table}

\begin{table}[!tb]
	\centering\footnotesize
	\caption{
		The summary statistics for Design 3 (High dimensional design) are represented for each setting of the true regression coefficients under the second compound symmetry covariance case. Different setting means different choice of the true coefficient $\beta_0$.}
	\vspace{.15cm}
	\begin{tabular}{c c c c c |   c c c c}
		\hline 
		& \multicolumn{4}{c}{ Setting 1 } & \multicolumn{4}{c}{ Setting 2 } \\ 
		& Sensitivity & Specificity & MCC & MSPE & Sensitivity & Specificity & MCC & MSPE \\ \hline
		GGS-G &0.23 &0.99 &0.27 &0.14 &0.20 &0.98 &0.21 &0.10   \\ 
		GGS-N  &0.33 &0.99 &0.39 &0.11 &0.32 &0.99 &0.31 &0.08 \\ 
		grLasso  &1 &0.81 &0.34 &0.14 &1 &0.81 &0.34 &0.12  \\ 
		grSCAD &1 &0.81 &0.34 &0.13 &1 &0.80  &0.32  &0.12  \\ 
		grMCP &0.33 &0.92 &0.15 &0.18 &0.33  &0.92 &0.15  &0.14  \\ 
		gel &0.33 &0.96 &0.23 &0.17 &0.33 &0.99 &0.39  &0.17 \\\hline \hline 
		& \multicolumn{4}{c}{ Setting 3 } & \multicolumn{4}{c}{ Setting 4 } \\ 
		& Sensitivity & Specificity & MCC & MSPE & Sensitivity & Specificity & MCC & MSPE \\ \hline
		GGS-G &0.26 &0.98 &0.27 &0.15 &0.27 &0.99 &0.31 &0.10   \\ 
		GGS-N &0.40 &1 &0.55 &0.07 &0.31 &0.99 &0.35 &0.10   \\ 
		grLasso  &1 &0.74 &0.28 &0.11&1 &0.78 &0.31 &0.11  \\ 
		grSCAD &1 &0.74 &0.28 &0.11 &1 &0.87 &0.40 &0.11 \\ 
		grMCP &0.33 &0.84 &0.08 &0.14 &0.33 &0.85 &0.08 &0.12 \\ 
		gel &0.33 &0.95 &0.20 &0.18&0.33 &0.95 &0.20 &0.18 \\\hline 
	\end{tabular}\label{table:comp8}
\end{table}

\begin{table}[!tb]
	\centering\footnotesize
	\caption{
		The summary statistics for Design 3 (High dimensional design) are represented for each setting of the true regression coefficients under the third autoregressive covariance case. Different setting means different choice of the true coefficient $\beta_0$.}
	\vspace{.15cm}
	\begin{tabular}{c c c c c |   c c c c}
		\hline 
		& \multicolumn{4}{c}{ Setting 1 } & \multicolumn{4}{c}{ Setting 2 } \\ 
		& Sensitivity & Specificity & MCC & MSPE & Sensitivity & Specificity & MCC & MSPE \\ \hline
		GGS-G &0.73 &1 &0.85 &0.10 &1 &1 &1 &0.09   \\ 
		GGS-N  &0.73 &1 &0.84 &0.10 &1 &1 &0.98 &0.10 \\ 
		grLasso  &1 &0.73 &0.28 &0.15 &1 &0.81 &0.34 &0.13  \\ 
		grSCAD &1 &0.81 &0.34 &0.15 &1 &0.86 &0.39 &0.12 \\ 
		grMCP &1 &0.81 &0.34 &0.13&1 &0.86 &0.39 &0.14  \\ 
		gel &0.67 &0.85 &0.23 &0.18 &1 &0.89 &0.44 &0.16 \\\hline \hline 
		& \multicolumn{4}{c}{ Setting 3 } & \multicolumn{4}{c}{ Setting 4 } \\ 
		& Sensitivity & Specificity & MCC & MSPE & Sensitivity & Specificity & MCC & MSPE \\ \hline
		GGS-G &1 &1 &1 &0.12 &1 &1 &1 &0.11   \\ 
		GGS-N &1 &1 &1 &0.12 &1 &1 &1 &0.11   \\ 
		grLasso  &1 &0.77 &0.30 &0.13 &1 &0.90 &0.45 &0.13  \\ 
		grSCAD &1 &0.90 &0.45 &0.12 &1 &0.90 &0.45 &0.12  \\ 
		grMCP &1 &0.86 &0.39 &0.13 &1 &0.90 &0.45 &0.11  \\ 
		gel &1 &0.86 &0.39 &0.14 &1 &0.87 &0.40 &0.14 \\\hline 
	\end{tabular}\label{table:comp9}
\end{table}

\newpage
Based on the above simulation results, we notice that in Design 1 (Baseline design) and Design 3 (High dimensional design) under the isotropic and the autoregressive covariance cases, our method overall works better than the regularization methods especially in strong signal settings (i.e., Settings 3 and 4). This is because as signal strength gets stronger, the consistency conditions of our method are easier to satisfy which leads to better performance. When the covariance is compound symmetric (i.e., Case 2) such that the correlations both within group and between group are equally 0.5, our method suffers from lower sensitivity compared with grLasso and grSCAD, but still has high specificity and MCC. For Design 2 (Dense model design), generally speaking, the proposed method is able to achieve better specificity and MCC, while the regularization methods have better sensitivity. It seems natural because the regularization methods based on cross-validation tend to include many redundant variables resulting in relatively larger number of  errors for the regularization methods compared with those for the Bayesian methods. Two different sampling algorithms GSS-G and GSS-N perform quite similar, where GSS-G works slightly better under the isotropic and the autoregressive covariance cases, while GSS-N obtains better results under the compound symmetric covariance structure. Overall, our simulation studies indicate that the proposed method can perform well under a variety of configurations with different dimensions, model complexities, sparsity levels, and correlation structures.

\section{Application to MRI Data Analysis}\label{sec:real}
We assess the effectiveness of the proposed approach on an MRI data set for predicting Parkinson's disease (PD). This study is approved by the Medical Research Ethical Committee of Nanjing Brain Hospital (Nanjing, China) in accordance with the Declaration of Helsinki with written informed consent  obtained from all subjects. 276 PD patients and 200 healthy controls (HCs) are recruited. Images are scanned on Siemens verio 3.0T superconducting MRI system with 8-channel  head coil in the department of radiology, and are available upon reasonable request. 

The original 3D T1-weighted images are then normalized using Statistical Parametric Mapping (\url{https://www.fil.ion.ucl.ac.uk/spm/software/spm12/}) on the Matlab platform. The detailed step includes spatial normalization to the Montreal Neurological Institute (MNI) space using the transformation parameters estimated via a unified segmentation algorithm \citep{Ashburner:2005} (Figure \ref{fig_1}). In particular, the unified segmentation algorithm adopts a probabilistic framework that enables image registration, tissue classification, and bias correction to be combined within the same generative model. Individual images for all subjects are therefore mapped from their individual MRI imaging space to a common reference space. As a result, the images of original size of (512, 512, 128) are converted into images of size (41, 48, 40) to reduce the complexity of the following analysis.
\begin{figure} [htbp]
	\begin{subfigure}[t]{.33\textwidth}
		\centering
		\includegraphics[width=1.0\linewidth]{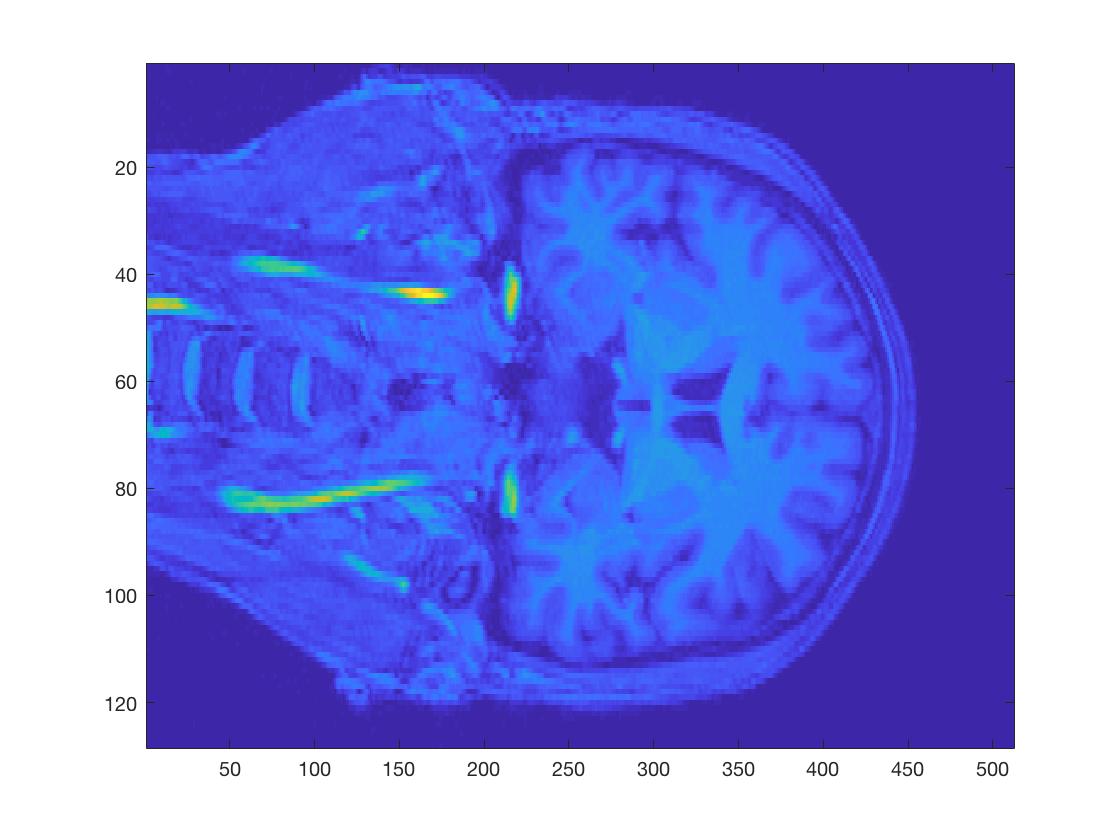}
		\caption{}
	\end{subfigure}
	\begin{subfigure}[t]{.33\textwidth}
		\centering
		\includegraphics[width=1.0\linewidth]{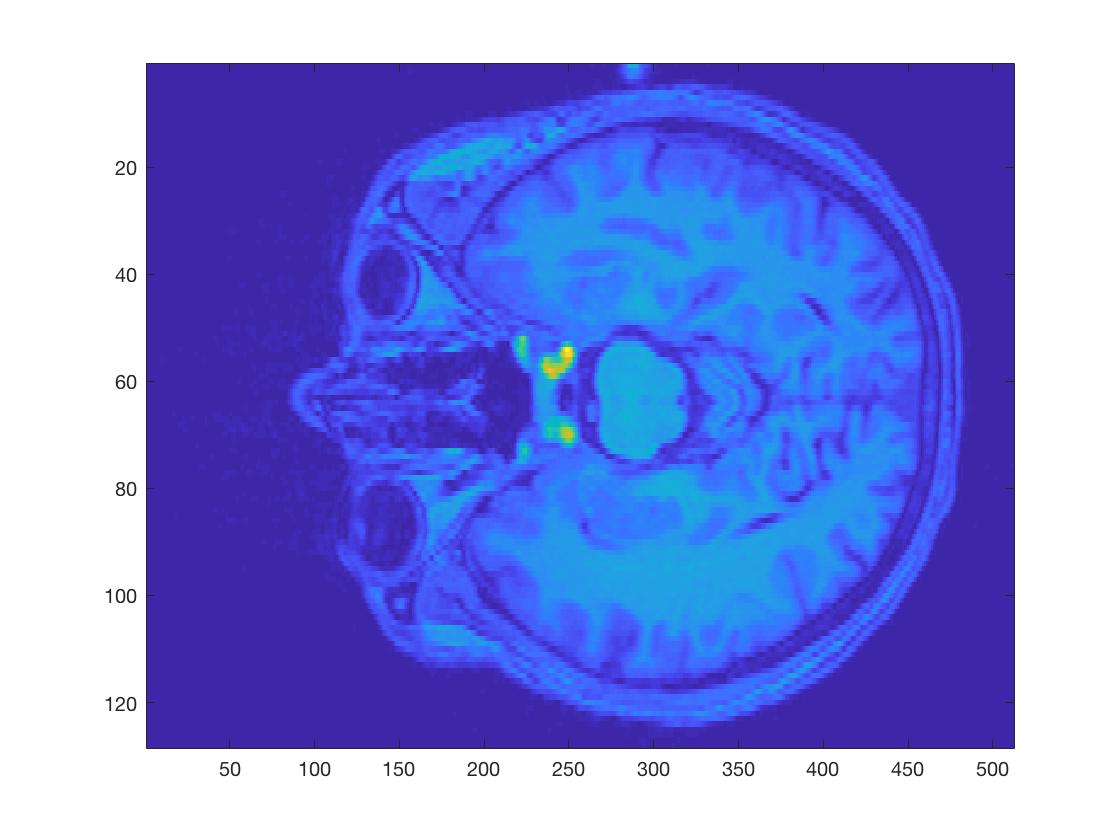}
		\caption{}
	\end{subfigure}
	\begin{subfigure}[t]{.33\textwidth}
		\centering
		\includegraphics[width=1.0\linewidth]{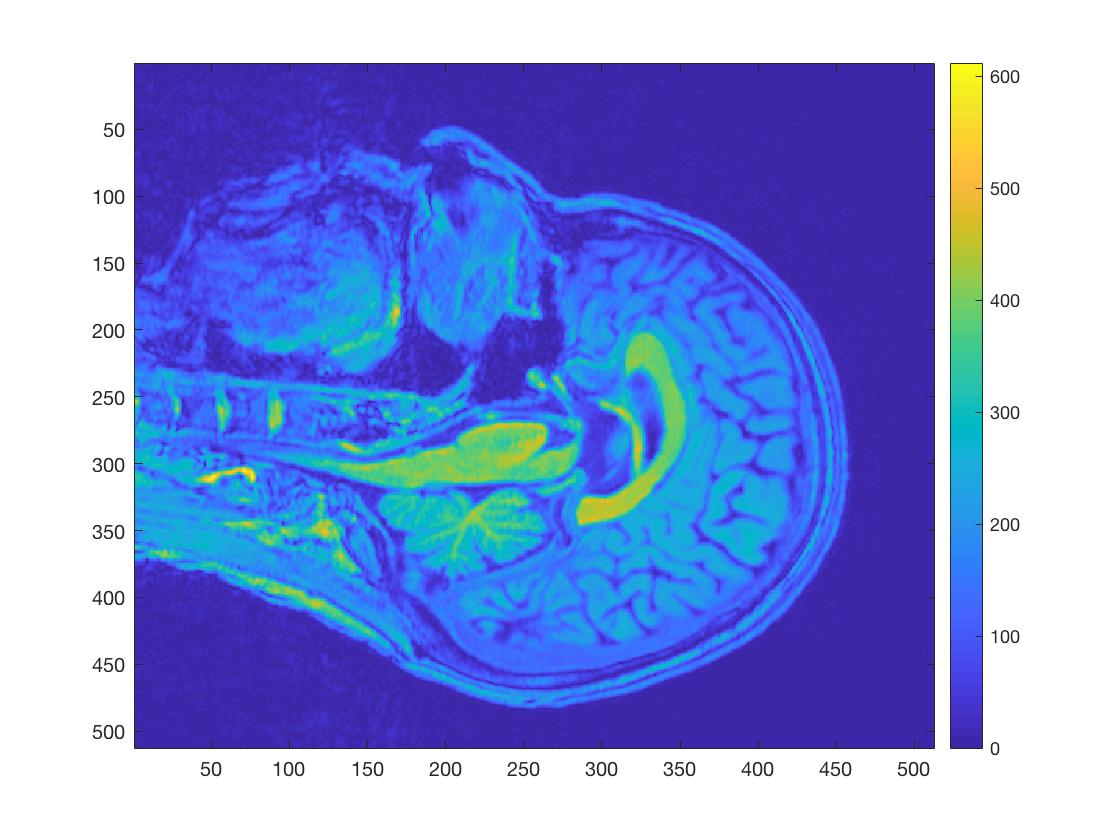}
		\caption{}
	\end{subfigure}
	\begin{subfigure}[t]{.33\textwidth}
		\centering
		\includegraphics[width=1.0\linewidth]{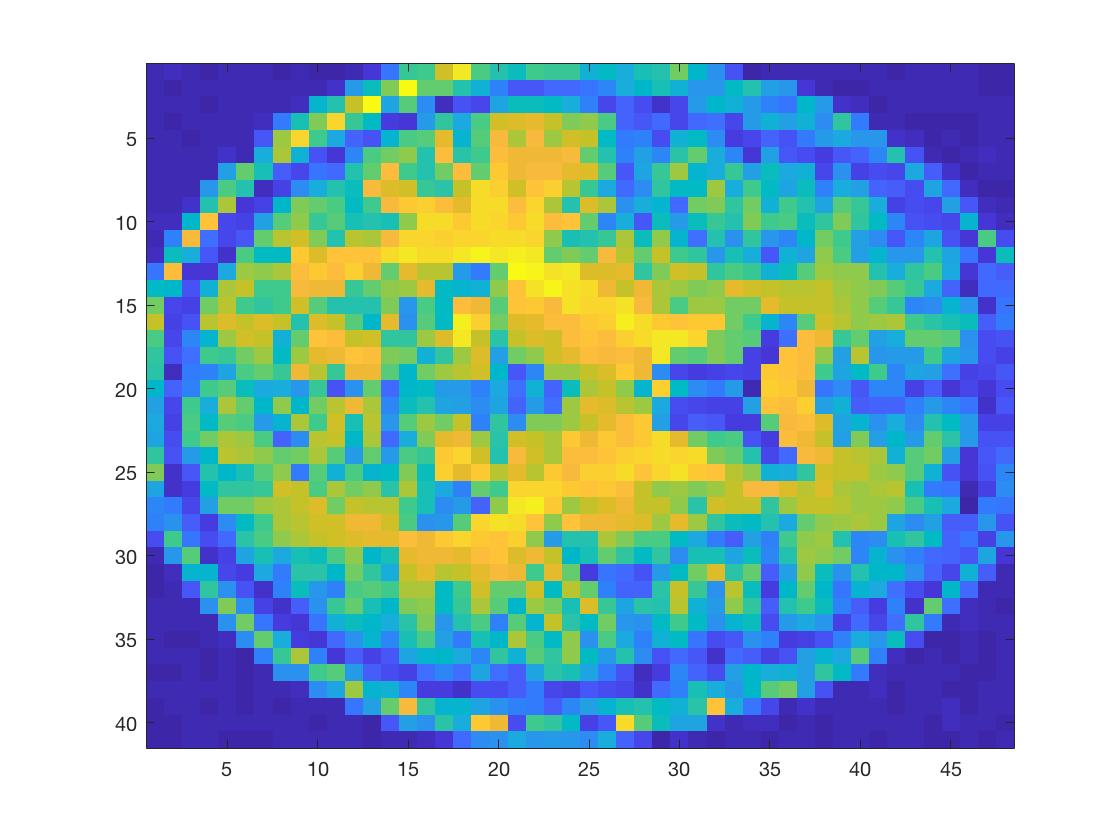}
		\caption{}
	\end{subfigure}
	\begin{subfigure}[t]{.33\textwidth}
		\centering
		\includegraphics[width=1.0\linewidth]{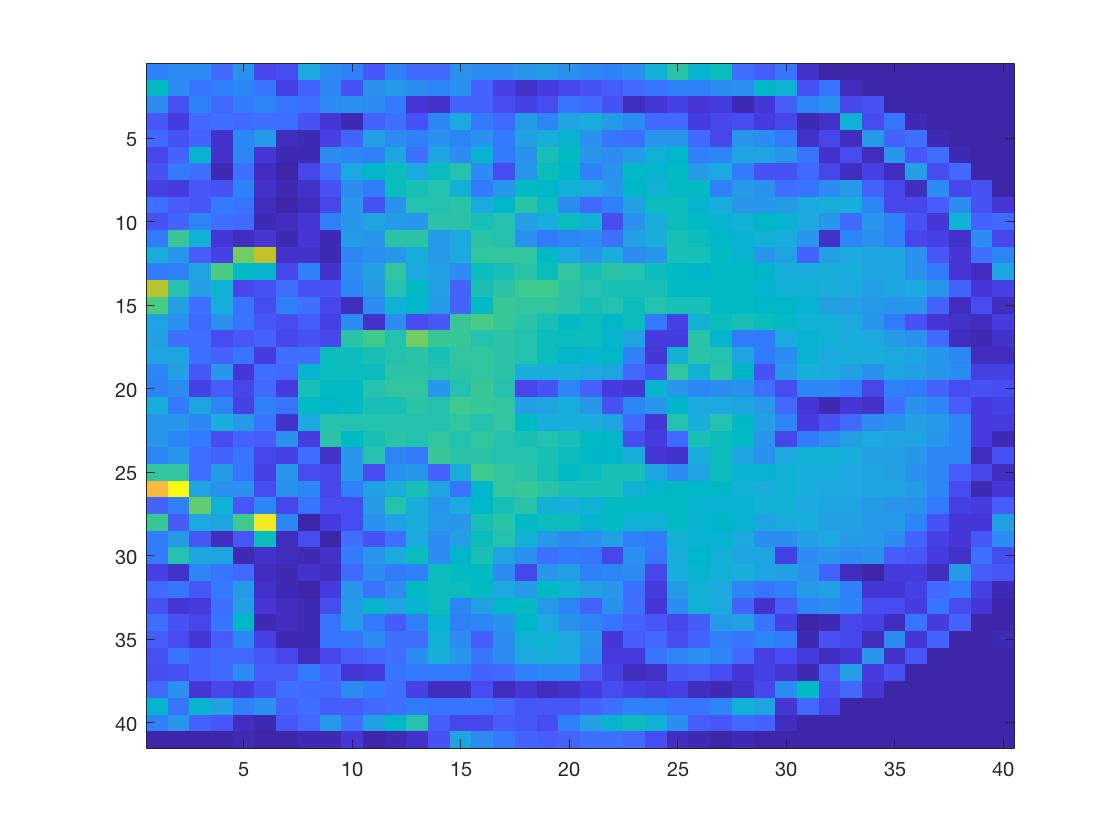}
		\caption{}
	\end{subfigure}
	\begin{subfigure}[t]{.33\textwidth}
		\centering
		\includegraphics[width=1.0\linewidth]{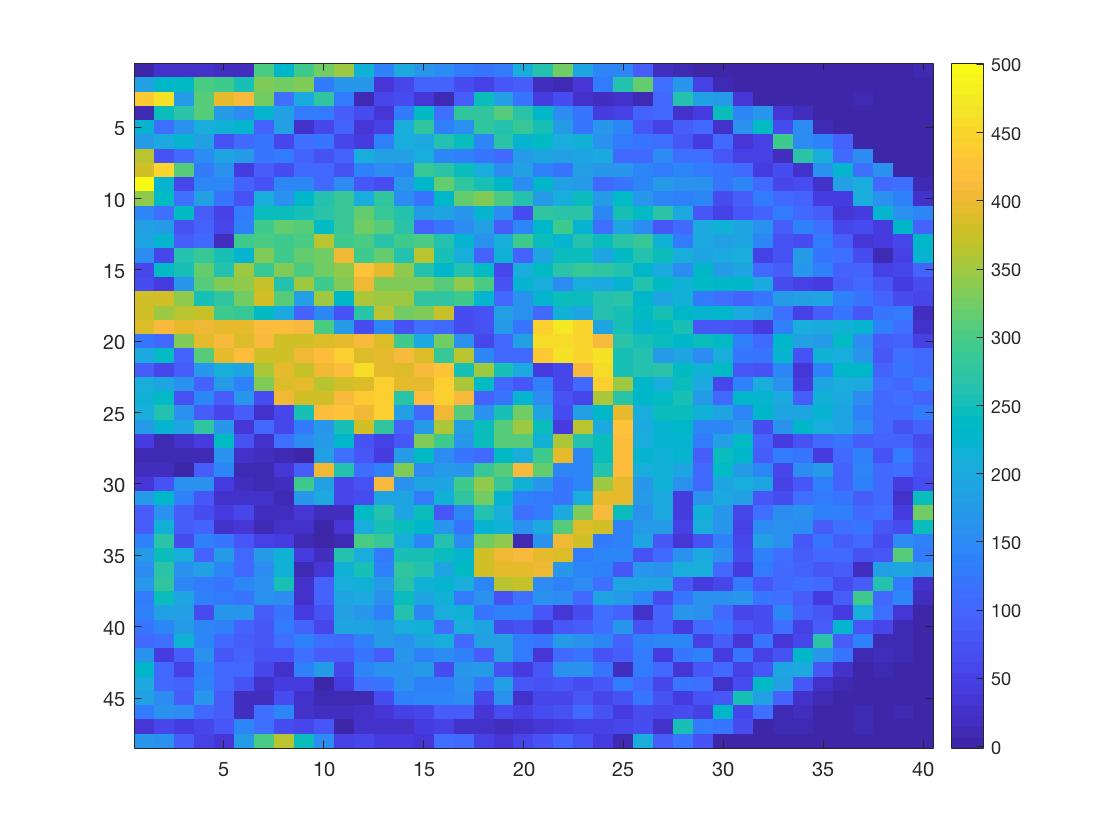}
		\caption{}
	\end{subfigure}
	\caption{Original MRI images (a-c) and normalized MRI images (d-f).}
	\label{fig_1}
\end{figure}

To identify significant brain regions for predicting PD, we first use the Harvard Oxford Atlas to define 13 brain regions across the whole ventral visual pathway and restrict our analysis to these brain regions, which resulted in 1064 voxels retained per image. For ease of implementation, we further divide these 1064 voxels into 108 groups, where each brain region can include several groups. Performance of GSS-G, GSS-N and other regularization methods are evaluated and tested on PD patients and healthy controls. The dataset was randomly divided into training set (80\%) and testing set (20\%) while retaining the PD:HC ratio in both sets. For Bayesian methods, we first obtain the estimated groups for GSS-G and GSS-N using methods specified previously, then evaluate the testing set performance using the standard glm estimates based on the selected groups. In Figure \ref{fig_2}, we draw the receiver operating characteristic (ROC) curves for all the methods. The results are further summarized in Table \ref{table:MRI} where a common cutoff value 0.5 is adopted for thresholding prediction. 

\begin{table}[tb]
	\centering\footnotesize
	\caption{
		The summary statistics for prediction performance on the testing set for all methods.}
	\vspace{.15cm}
	\begin{tabular}{c c c c c}
		\hline 
		& Sensitivity & Specificity & MCC & MSPE  \\ \hline
		GGS-G &0.89 &0.73 &0.63 &0.14   \\ 
		GGS-N  &0.89 &0.57 &0.46 &0.17 \\ 
		grLasso  &0.87 &0.50 &0.43 &0.17  \\ 
		grSCAD &0.86 &0.45 &0.39 &0.18   \\ 
		grMCP &0.86 &0.64 &0.51 &0.20  \\ 
		gel &0.86 &0.55 &0.43 &0.15  \\\hline 
	\end{tabular}\label{table:MRI}
\end{table}

\begin{figure}
	\centering
	\includegraphics[width=0.65\linewidth]{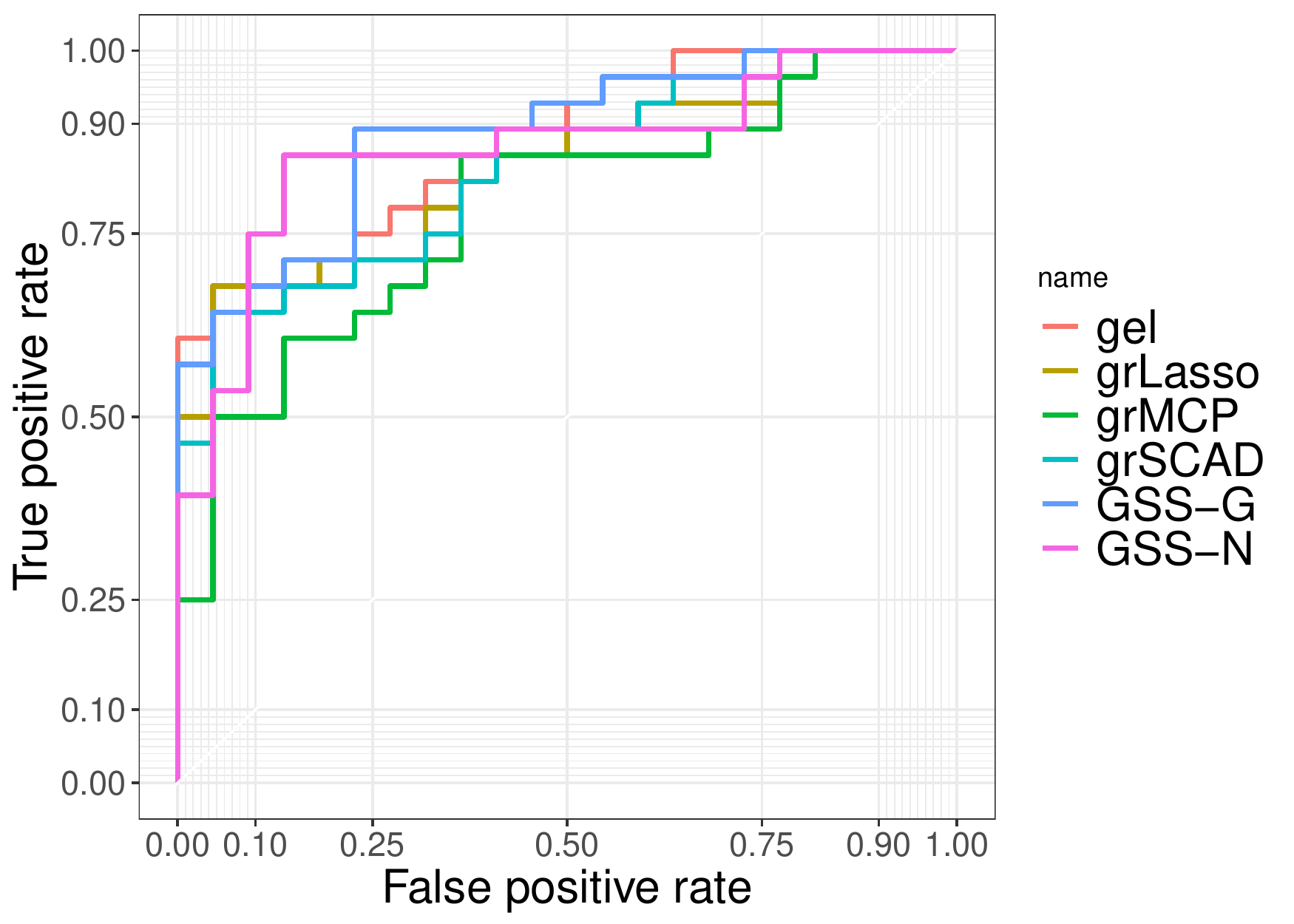}
	\caption{ROC curves comparison between different methods. X-axis: false positive rate (1 - specificity). Y-axis: true positive rate (sensitivity).}
	\label{fig_2}
\end{figure}
From Table \ref{table:MRI} and Figure \ref{fig_2}, it seems that GGS-G has overall better prediction performance compared with other methods. GGS-N has higher sensitivity but lower specificity compared with GGS-G and grMCP. Both GGS-G and GGS-N are able to identify aberrant structures for PD that reside in regions of interest including cingulate gyrus, frontal medial cortex, inferior frontal gyrus, parahippocampal gyrus, brainstem, left hippocampus and left thalamus. These findings suggest disease-related alterations of  structure as the basis for faulty information processing in this disorder. Our findings are in good agreement with the alterative structure and functional features in cortical regions, cerebellum, brainstem, thalamus  and limbic regions in previous studies \citep{Wei2012,Zhu2016,Ling:1,Bing:3,Bing:4,Xuan:7}.

\section{Discussion}\label{sec:disc}
We focus on group selection for logistic regression models in this paper, but often a bi-level selection can be of significant interest, where we further aim to select nonzero coefficients among selected groups.
In this case, we can extend the group spike and slab prior in \eqref{model1} and \eqref{model2} by modifying the slab part, $N_{|G_j|}(\beta_{G_j} \mid 0, \tau^2 I_{|G_j|})$, as follows:
\bea
\sum_{l \in G_j} \Big\{  (1-Z_{jl}) \delta_0(\beta_{G_j, l}) + Z_{jl} N(\beta_{G_j, l} \mid 0, \tau^2)   \Big\} ,
\eea
where $Z_{jl} \overset{i.i.d.}{\sim} {\rm Bernoulli}( q^* )$ for $l \in G_j$ and some $0< q^*<1$.
\cite{Yang:Naveen:2018} also mentioned a similar extension.
However, it will increase the computational burden due to the larger model space, and it is not sure whether theoretical results in this paper can be obtained from this extended model.
For this reason, we focused on the group selection problem in this paper and leave the bi-level selection problem for future research.

\newpage

\section*{Supplementary Material}

Throughout the Supplementary Material, we assume that for any  
$$u \in \{u \in \bbR^n : u \text{ is in the space spanned by the columns of } \sg^{1/2} X_{G_k} \}$$ 
and any model $k \in \{ k \subseteq [r] : |G_k| \le m_n + |G_t| \}$, there exists $\delta^*>0$ such that
\bean\label{delta_star}
\bbE \Big[ \exp \big\{ u^T \sg^{-1/2}(E - \mu) \big\} \Big] &\le& \exp\Big\{  \frac{(1+\delta^*) u^T u}{2} \Big\} ,
\eean
for any $n \ge N(\delta^*)$.
However, as stated in \cite{Naveen:2018}, there always exists $\delta^*>0$ satisfying inequality \eqref{delta_star}, so it is not really a restriction.
Since we will focus on sufficiently large $n$, $\delta^*$ can be considered an arbitrarily small constant, so  we can always assume  that $\delta > \delta^*$.

\begin{proof}[Proof of Theorem \ref{thm:nosuper}]
	Let $M_1 = \{ k: k \supsetneq t, |G_k| \le m_n \}$ and
	\bea
	PR(k,t) &=& \frac{\pi(Z=k \mid E)}{\pi(Z=t \mid E)},
	\eea
	where $t \subseteq [r]$ is the true model.
	We will show that 
	\bean\label{sel_goal1}
	\sum_{k: k \in M_1} PR(k,t) &\overset{P}{\lra}& 0 \quad \text{ as } n\to\infty.
	\eean
	
	By Taylor's expansion of $L_ n(\beta_{G_k})$ around $\what{\beta}_{G_k}$, which is the MLE of $\beta_{G_k}$ under the model $k$, we have
	\bea
	L_n(\beta_{G_k}) &=& L_n (\what{\beta}_{G_k}) - \frac{1}{2} (\beta_{G_k} - \what{\beta}_{G_k} )^T H_n(\tilde{\beta}_{G_k}) (\beta_{G_k} - \what{\beta}_{G_k} ) 
	\eea
	for some $\tilde{\beta}_{G_k}$ such that $\|\tilde{\beta}_{G_k} - \what{\beta}_{G_k}\|_2 \le \|\beta_{G_k} - \what{\beta}_{G_k}\|_2$.
	Furthermore, by Lemmas A.1 and A.3 in \cite{Naveen:2018} and Condition \hyperref[cond_A2]{\rm (A2)}, with probability tending to 1,
	\bea
	L_n(\beta_{G_k}) - L_n (\what{\beta}_{G_k}) 
	&\le& - \frac{1-\epsilon}{2} (\beta_{G_k} - \what{\beta}_{G_k} )^T H_n(\beta_{0,G_k}) (\beta_{G_k} - \what{\beta}_{G_k} ) 
	\eea
	for any $k\in M_1$ and $\beta_{G_k}$ such that $\|\beta_{G_k} - \beta_{0, G_k}\|_2 < c \sqrt{ |G_k| \Lambda_{|k|} \log r /n } = : c w_n$, where $\epsilon = \epsilon_n := c'\sqrt{m_n^2 \Lambda_{m_n} \log r /n } = o(1)$, for some constants $c,c'>0$.
	Note that for $\beta_{G_k}$ such that $\|\beta_{G_k} - \what{\beta}_{G_k}\|_2 = cw_n /2$, 
	\bea
	L_n(\beta_{G_k}) - L_n (\what{\beta}_{G_k})  
	&\le& - \frac{1-\epsilon}{2} \, \|\beta_{G_k} - \what{\beta}_{G_k}\|_2^2 \, \lambda_{\min} \big\{ H_n(\beta_{0,G_k})  \big\} \\
	&\le& - \frac{1-\epsilon}{2} \, \frac{c^2 w_n^2}{4} n \lambda \\
	&=& - \frac{1-\epsilon}{8} c^2 \lambda |G_k| \Lambda_{|k|} \log r \,\, \lra \,\, - \infty \quad \text{ as }n\to\infty,
	\eea
	where the second inequality holds due to Condition \hyperref[cond_A2]{\rm (A2)}.
	It also holds for any $\beta_{G_k}$ such that $\|\beta_{G_k} - \what{\beta}_{G_k}\|_2 > cw_n /2$ by concavity of $L_n(\cdot)$ and the fact that $\what{\beta}_{G_k}$ maximizes $L_n(\beta_{G_k})$.
	
	Define the set
	\bea
	B &:=& \big\{ \beta_{G_k}:  \|\beta_{G_k} - \what{\beta}_{G_k}\|_2 \le cw_n /2  \big\},
	\eea
	then we have $B \subset \{\beta_{G_k}: \|\beta_{G_k}- \beta_{0,G_k}\|_2 \le cw_n  \}$ for some large $c>0$ and any $k\in M_1$, with probability tending to 1.
	Therefore, with probability tending to 1, 
	\bea
	&&\int \exp \big\{ L_n(\beta_{G_k}) \big\} (2\pi \tau^2)^{-\frac{|G_k|}{2}} \exp \big\{ - \frac{1}{2\tau^2}\|\beta_{G_k}\|_2^2 \big\} d \beta_{G_k} \\
	&\le& (2\pi \tau^2)^{-\frac{|G_k|}{2}}  \exp \big\{ L_n(\what{\beta}_{G_k}) \big\} \Big[ \int_B \exp \big\{ -\frac{1-\epsilon}{2} (\beta_{G_k} - \what{\beta}_{G_k} )^T H_n(\beta_{0,G_k}) (\beta_{G_k} - \what{\beta}_{G_k} )  - \frac{1}{2\tau^2}\|\beta_{G_k}\|_2^2  \big\} d\beta_{G_k} \\
	&& + \,\, \exp \big\{ -\frac{1-\epsilon}{8} c^2 \lambda |G_k| \Lambda_{|k|} \log r  \big\} \int_{B^c} \exp \big\{- \frac{1}{2\tau^2}\|\beta_{G_k}\|_2^2  \big\} d\beta_{G_k}  \Big] \\
	&\le&  C(\tau^2)^{-\frac{|G_k|}{2}}  \exp \big\{ L_n(\what{\beta}_{G_k}) \big\} \det \Big\{ H_n(\beta_{0,G_k}) (1-\epsilon) + \tau^{-2}I  \Big\}^{-1/2}    \big( 1+ o(1) \big)
	\eea
	for any $k\in  M_1$ and some constant $C>0$ because, for $A_k = (1-\epsilon)H_n(\beta_{0,G_k})$ and $\beta_{G_k}^* = (A_k+\tau^{-2} I)^{-1} A_k \what{\beta}_{G_k}$,
	\bea
	&& \int_B \exp \big\{ -\frac{1}{2} (\beta_{G_k} - \what{\beta}_{G_k} )^T A_k (\beta_{G_k} - \what{\beta}_{G_k} )  - \frac{1}{2\tau^2}\|\beta_{G_k}\|_2^2  \big\} d\beta_{G_k}  \\
	&\le& \int \exp \big\{ -\frac{1}{2} (\beta_{G_k} - {\beta}^*_{G_k} )^T (A_k+ \tau^{-2} I) (\beta_{G_k} - {\beta}^*_{G_k} )    \big\} d\beta_{G_k} \\
	&&\times \,\, \exp \big\{ -\frac{1}{2} \what{\beta}_{G_k}^T \big( A_k - A_k (A_k+\tau^{-2}I)^{-1} A_k \big)\what{\beta}_{G_k}   \big\}  \\
	&=&  (2\pi)^{|G_k|/2} \det \big( A_k+ \tau^{-2} I \big)^{-1/2} \, \exp \big\{ -\frac{1}{2} \what{\beta}_{G_k}^T \big( A_k - A_k(A_k+\tau^{-2}I)^{-1} A_k \big)\what{\beta}_{G_k}   \big\} \\
	&\le& (2\pi)^{|G_k|/2} \det \big( A_k+ \tau^{-2} I \big)^{-1/2} 
	\eea
	and 
	\bea
	\det \Big\{ H_n(\beta_{0,G_k}) (1-\epsilon) + \tau^{-2}I  \Big\}^{1/2}
	&\le&  \big(  n\Lambda_{|k|} + \tau^{-2}  \big)^{|G_k|/2}  \\
	&\le& \exp \big\{  C |G_k| \log n  \big\}  \\
	&\ll&\exp \big\{ \frac{1-\epsilon}{8} c^2 \lambda |G_k| \Lambda_{|k|} \log r  \big\} 
	\eea
	for some constant $C>0$ and some large constant $c>0$, by Conditions \hyperref[cond_A1]{\rm (A1)}, \hyperref[cond_A2]{\rm (A2)} and \hyperref[cond_A4]{\rm (A4)}.
	Thus, with probability tending to 1,
	\bea
	\pi(Z=k \mid E) &=& C^* \Big(\frac{q}{1-q} \Big)^{|k|} \int \exp \big\{ L_n(\beta_{G_k}) \big\} (2\pi \tau^2)^{-\frac{|G_k|}{2}} \exp \big\{ - \frac{1}{2\tau^2}\|\beta_{G_k}\|_2^2 \big\} d \beta_{G_k} \\
	&\lesssim& C^* \Big(\frac{q}{1-q} \Big)^{|k|}  (\tau^2)^{-\frac{|G_k|}{2}} \exp \big\{ L_n(\what{\beta}_{G_k}) \big\} \det \Big\{ H_n(\beta_{0,G_k}) (1-\epsilon) + \tau^{-2}I  \Big\}^{-1/2}   \big( 1+ o(1) \big) 
	\eea
	for any $k \in M_1$, where $C^*>0$ is a normalizing constant.
	Similarly, with probability tending to 1, we have
	\bea
	\pi(Z = t \mid E) &\gtrsim& C^* \Big(\frac{q}{1-q} \Big)^{|t|}  (\tau^2)^{-\frac{|G_t|}{2}}  \exp \big\{ L_n(\what{\beta}_{G_t}) \big\} \det \Big\{ H_n(\beta_{0,G_t}) (1+\epsilon) + \tau^{-2}I  \Big\}^{-1/2} \\
	&& \times \,\,  \exp \big\{ -\frac{1}{2} \what{\beta}_{G_t}^T \big( A_t - A_t(A_t+\tau^{-2}I)^{-1} A_t \big)\what{\beta}_{G_t}   \big\} \\
	&\gtrsim&  C^* \Big(\frac{q}{1-q} \Big)^{|t|}  (\tau^2)^{-\frac{|G_t|}{2}}  \exp \big\{ L_n(\what{\beta}_{G_t}) \big\} \det \Big\{ H_n(\beta_{0,G_t}) (1+\epsilon) + \tau^{-2}I  \Big\}^{-1/2}
	\eea
	by Lemma \ref{lem_aux1}, where $A_t = (1+\epsilon)H_n(\beta_{0,G_t})$.
	Therefore, with probability tending to 1, 
	\bea
	\frac{\pi(Z = k \mid E) }{\pi(Z = t \mid E) }
	&\lesssim& \Big(\frac{q}{1-q} \Big)^{|k|-|t|}  (n\tau^2)^{-\frac{|G_k|-|G_t|}{2}}  \frac{\det \Big\{ n^{-1} H_n(\beta_{0,G_t}) (1+\epsilon) + (n\tau)^{-2}I  \Big\}^{1/2} }{\det \Big\{ n^{-1}H_n(\beta_{0,G_k}) (1-\epsilon) + (n\tau)^{-2}I  \Big\}^{1/2} } \\
	&& \times \,\, \exp \big\{ L_n(\what{\beta}_{G_k}) - L_n(\what{\beta}_{G_t})  \big\} \\
	&\lesssim& \Big( \frac{2}{r} \Big)^{|k|-|t|}   \Big( \frac{\lambda n\tau^2}{2} \Big)^{-\frac{|G_k|-|G_t|}{2}}  \exp \big\{ L_n(\what{\beta}_{G_k}) - L_n(\what{\beta}_{G_t})  \big\} 
	\eea
	for any $k \in M_1$, where the second inequality holds by Lemma \ref{lem_aux2}.
	Finally, by applying Lemma \ref{lem_aux3}, we obtain
	\bean\label{loglike_diff}
	L_n(\what{\beta}_{G_k}) - L_n(\what{\beta}_{G_t})  
	&\le& b_n (|G_k| - |G_t|)
	\eean
	for any $k \in M_1$ with probability tending to 1, where $b_n = (1+\delta^*)(1+2w) \log r$ such that $1+\delta > (1+\delta^*)(1+2w)$.
	
	Hence, with probability tending to 1,
	\bea
	\sum_{k\in M_1} PR(k,t) 
	&\lesssim& \sum_{k\in M_1} \Big( \frac{2}{r} \Big)^{|k|-|t|}   \Big( \frac{\lambda n\tau^2}{2} \Big)^{-\frac{|G_k|-|G_t|}{2}}  \exp \big\{  b_n(|G_k|-|G_t|)  \big\}  \\
	&\lesssim& \sum_{d=|t|+1}^{r}  \sum_{k: k \supsetneq t, |k| = d}\Big( \frac{2}{r} \Big)^{d-|t|}  r^{-(|G_k|-|G_t|)(1+\delta)} r^{(|G_k|-|G_t|)(1+\delta^*)(1+2w)} \\
	&\lesssim& \sum_{d=|t|+1}^{r}  \sum_{k: k \supsetneq t, |k| = d} \Big( \frac{2}{r} \Big)^{d-|t|}  r^{-C(|G_k|-|G_t|)} \\
	&\lesssim& \sum_{d=|t|+1}^{r}  \binom{r-|t|}{d-|t|}  \Big( \frac{2}{r} \Big)^{d-|t|}  r^{-C(d-|t|)} \\
	&\lesssim& \sum_{d=|t|+1}^{r}  (C'r)^{- C(d -|t|)} \,\,=\,\, o(1)
	\eea
	for some constants $C , C'>0$.
	Thus, we have proved the desired result \eqref{sel_goal1}.
\end{proof}

\begin{proof}[Proof of Theorem \ref{thm:ratio}]
	Let $M_2 = \{k : k \nsupseteq t,  |G_k|\le m_n \}$.
	For any $k \in M_2$, let $k^* = k \cup t$, so that $k^* \in M_1$.
	Let $\beta_{G_{k^*}} $ be the $|G_{k^*}|$-dimensional vector including $\beta_{G_k}$ for $G_k$ and zeros for $G_{t \setminus k}$.
	Then by Taylor's expansion and Lemmas A.1 and A.3 in \cite{Naveen:2018}, with probability tending to 1,
	\bea
	L_n (\beta_{G_{k^*}}) 
	&=& L_n ( \what{\beta}_{G_{k^*}} ) - \frac{1}{2} ( \beta_{G_{k^*}} - \what{\beta}_{G_{k^*}} )^T H_n( \tilde{\beta}_{G_{k^*}})( \beta_{G_{k^*}} - \what{\beta}_{G_{k^*}} ) \\
	&\le& L_n ( \what{\beta}_{G_{k^*}} ) - \frac{1-\epsilon}{2} ( \beta_{G_{k^*}} - \what{\beta}_{G_{k^*}} )^T H_n( {\beta}_{0,G_{k^*}})( \beta_{G_{k^*}} - \what{\beta}_{G_{k^*}} )  \\
	&\le& L_n ( \what{\beta}_{G_{k^*}} ) - \frac{n (1-\epsilon)\lambda}{2} \|\beta_{G_{k^*}} - \what{\beta}_{G_{k^*}}\|_2^2
	\eea
	for any $\beta_{G_{k^*}}$ such that $\| \beta_{G_{k^*}} - \beta_{0,G_{k^*}} \|_2 \le c \sqrt{ |G_{k^*}| \Lambda_{|k^*|} \log r /n } = c w_n$ for some large constant $c>0$.
	Note that
	\bea
	&& \int \exp \Big\{  - \frac{n(1-\epsilon)\lambda}{2} \| \beta_{G_{k^*}} - \what{\beta}_{G_{k^*}} \|_2^2 - \frac{1}{2\tau^2} \|\beta_{G_k}\|_2^2  \Big\} d\beta_{G_k} \\
	&=& \int \exp \Big\{  - \frac{n(1-\epsilon)\lambda}{2} \| \beta_{G_{k}} - \what{\beta}_{G_{k}} \|_2^2 - \frac{1}{2\tau^2} \|\beta_{G_k}\|_2^2  \Big\} d\beta_{G_k} \, \exp \Big\{ - \frac{n(1-\epsilon)\lambda}{2} \|\what{\beta}_{G_{t\setminus k}} \|_2^2  \Big\} \\
	&=& (2\pi)^{\frac{|G_k|}{2}} \big\{ n(1-\epsilon)\lambda  + \tau^{-2}\big\}^{-\frac{|G_k|}{2}} \exp \Big\{ - \frac{n(1-\epsilon)\lambda}{2} \|\what{\beta}_{G_{t\setminus k}} \|_2^2  \Big\}.
	\eea
	Define the set
	\bea
	B_* &:=& \big\{ \beta_{G_k}:  \|\beta_{G_{k^*}} - \what{\beta}_{G_{k^*}}\|_2 \le cw_n /2  \big\},
	\eea
	for some large constant $c>0$, then by similar arguments used for super sets, with probability tending to 1,
	\bea
	&& \pi(Z= k\mid E) \\
	&=&  C^* \Big(\frac{q}{1-q} \Big)^{|k|} \int_{B_* \cup B_*^c} \exp \big\{ L_n(\beta_{G_{k^*}}) \big\} (2\pi \tau^2)^{-\frac{|G_k|}{2}} \exp \big\{ - \frac{1}{2\tau^2}\|\beta_{G_k}\|_2^2 \big\} d \beta_{G_k}  \\
	&\lesssim& C^*\Big(\frac{q}{1-q} \Big)^{|k|} \exp \big\{  L_n(\what{\beta}_{G_{k^*}})  \big\} (\tau^2)^{-\frac{|G_k|}{2}} \\
	&&\times \,\, \Big[  \big\{ n(1-\epsilon)\lambda  + \tau^{-2}\big\}^{-\frac{|G_k|}{2}} \exp \Big\{ - \frac{n(1-\epsilon)\lambda}{2} \|\what{\beta}_{G_{t\setminus k}} \|_2^2  \Big\}  +  \exp \big\{  - c C |G_{k^*}| \Lambda_{|k^*|} \log r \big\} \Big]
	\eea
	for any $k \in M_2$ and for some constant $C>0$.
	Since the lower bound for $\pi(Z= t\mid E)$ can be derived as before, it leads to
	\bean
	&& PR(k,t) \nonumber\\
	&\lesssim& 
	\Big(\frac{q}{1-q} \Big)^{|k|- |t|} (n \tau^2)^{- \frac{|G_k|- |G_t|}{2}} 
	\frac{ \det \big\{ (1+\epsilon)n^{-1} H_n(\beta_{0,G_t}) + (n\tau^2)^{-1} I  \big\}^{1/2} }{ \big\{ (1-\epsilon)\lambda + (n\tau^2)^{-1} \big\}^{|G_k|/2} } \nonumber \\
	&&\times \,\, \exp \big\{ L_n(\what{\beta}_{G_{k^*}}) - L_n(\what{\beta}_{G_t})  \big\}  \,  \exp \Big\{ - \frac{n(1-\epsilon)\lambda}{2} \|\what{\beta}_{G_{t\setminus k}} \|_2^2  \Big\} \label{M2_part1}\\
	&+& \Big(\frac{q}{1-q} \Big)^{|k|- |t|} (n \tau^2)^{- \frac{|G_k|- |G_t|}{2}}   \det \big\{ (1+\epsilon)n^{-1} H_n(\beta_{0,G_t}) + (n\tau^2)^{-1} I  \big\}^{1/2}  \exp \big\{ L_n(\what{\beta}_{G_{k^*}}) - L_n(\what{\beta}_{G_t})  \big\} \nonumber \\
	&&\times \,\, \exp \big\{  - c\, C |G_{k^*}| \Lambda_{|k^*|} \log r \big\}  \label{M2_part2}
	\eean
	for any $k \in M_2$ with probability tending to 1.
	
	We first focus on \eqref{M2_part1}.
	Note that
	\bea
	&& \frac{ \det \big\{ (1+\epsilon)n^{-1} H_n(\beta_{0,G_t}) + (n\tau^2)^{-1} I  \big\}^{1/2} }{ \big\{ (1-\epsilon)\lambda + (n\tau^2)^{-1} \big\}^{|G_k|/2} }\\
	&\le& \frac{  \big\{ (1+\epsilon) \Lambda_{|t|}  + (n\tau^2)^{-1}  \big\}^{|G_t|/2} }{ \big\{ (1-\epsilon)\lambda + (n\tau^2)^{-1} \big\}^{|G_k|/2} } \\
	&=&  \Big\{  \frac{(1+\epsilon) \Lambda_{|t|}  + (n\tau^2)^{-1} }{(1-\epsilon)\lambda + (n\tau^2)^{-1}}  \Big\}^{|G_t|/2}   \Big\{ \frac{1}{(1-\epsilon)\lambda + (n\tau^2)^{-1} } \Big\}^{(|G_k|-|G_t|)/2} \\
	&\lesssim& \exp \big\{ |G_t| \log  \Lambda_{|t|}   \big\} \, \Big\{ \frac{1}{(1-\epsilon)\lambda + (n\tau^2)^{-1} } \Big\}^{(|G_k|-|G_t|)/2}
	\eea
	for some constant $C>0$.
	Furthermore, by the same arguments used in \eqref{loglike_diff}, we have
	\bea
	L_n(\what{\beta}_{G_{k^*}}) - L_n(\what{\beta}_{G_t})  
	&\lesssim& C_* \log r (|G_{k^*}| - |G_t|) \\
	&=& C_* \log r |G_{t\setminus k}| +  C_* \log r ( |G_k|- |G_t| )
	\eea
	for some constant $C_* >0$ and for any $k \in M_2$ with probability tending to 1.
	Here we choose $C_* = (1+\delta^*)(1+2w)$ if $|G_k|> |G_t|$ or $C_* = 2+\delta$ if $|G_k| \le |G_t|$ so that
	\bea
	(n \tau^2)^{- \frac{|G_k|- |G_t|}{2}} \Big\{ \frac{1}{(1-\epsilon)\lambda + (n\tau^2)^{-1} } \Big\}^{\frac{|G_k|- |G_t|}{2}}  r^{ C_* ( |G_k|- |G_t| )} 
	&\lesssim& r^{ (C_* - 1 - \delta) ( |G_k|- |G_t| ) } \,\,=\,\,o(1),
	\eea
	where the inequality holds by Condition \hyperref[cond_A4]{\rm (A4)}.
	To be more specific, we divide $M_2$ into two disjoint sets $M_2' = \{ k : k \in M_2 ,  |G_t| < |G_k| \le m_n \}$ and $M_2^* = \{ k : k \in M_2 , |G_k| \le |G_t| \}$, and will show that $\sum_{k \in M_2'} PR(k,t)  + \sum_{k \in M_2^*} PR(k,t)\lra 0$ as $n\to\infty$ with probability tending to 1.
	Thus, we can choose different $C_*$ for $M_2'$ and $M_2^*$ as long as $C_* \ge (1+\delta^*)(1+2w)$.	
	On the other hand, with probability tending to 1, by Condition \hyperref[cond_A3]{\rm (A3)},
	\bea
	\exp \Big\{ - \frac{n(1-\epsilon)\lambda}{2} \|\what{\beta}_{G_{t\setminus k}} \|_2^2  \Big\} 
	&\le&  \exp \Big[ - \frac{n(1-\epsilon)\lambda}{2}  \Big\{  \| \beta_{0,G_{t\setminus k}} \|_2^2 - \|\what{\beta}_{G_{t\setminus k}} - \beta_{0,G_{t\setminus k}} \|_2^2 \Big\}  \Big]  \\
	&\le&  \exp \Big[ - \frac{n(1-\epsilon)\lambda}{2}  \Big\{  |t\setminus k|^2 \min_{j \in t}\|\beta_{0,G_j}\|_2^2 - c' w_n'^2 \Big\}  \Big]  \\
	&\le& \exp \Big\{ - \frac{(1-\epsilon)\lambda}{2}   \big( c_0 |t\setminus k|^2 |G_t|    - c' |G_{t\setminus k}|  \big)  \Lambda_{|t|} \log r   \Big\}  \\
	&\le& \exp \Big\{  - \frac{(1-\epsilon)\lambda}{2}  ( c_0 - c') |t\setminus k|^2 |G_t|  \Lambda_{|t|} \log r    \Big\}
	\eea
	for any $k \in M_2$ and some large constants $c_0> c'>0$, where $w_n'^2 = |G_{t\setminus k}| \Lambda_{|t\setminus k|} \log r/n $.
	Here, $c'  = 5 \lambda^{-2} (1-\epsilon)^{-2}$ by the proof of Lemma A.3 in \cite{Naveen:2018}.
	
	Hence,  \eqref{M2_part1} for any $k \in M_2$ is bounded above by
	\bea
	&& \Big( \frac{2}{r} \Big)^{|k|-|t|} \exp \Big\{  |G_t| \log \Lambda_{|t|} + C_* |G_{t\setminus k}| \log r - \frac{(1-\epsilon)\lambda}{2}   (c_0- c') |t\setminus k|^2 |G_t|  \Lambda_{|t|} \log r    \Big\}   \\
	&\lesssim&    \exp \Big\{ - \Big( \frac{(1-\epsilon)\lambda}{2}  (c_0- c')  - C_* - o(1) \Big) |t\setminus k|^2 |G_t|  \Lambda_{|t|} \log r  + |t| \log r  \Big\} \\
	&\le&    \exp \Big\{ - \Big( \frac{(1-\epsilon)\lambda}{2}  (c_0- c')  - C_* - 1 - o(1) \Big) |t\setminus k|^2 |G_t|  \Lambda_{|t|} \log r  \Big\} \\
	&\le&  \exp \Big\{ - \Big( \frac{(1-\epsilon)\lambda}{2}  (c_0- c')  - C_* - 1 - o(1) \Big)  |G_t|  \Lambda_{|t|} \log r  \Big\} 
	\eea
	with probability tending to 1, where the last term is of order $o(1)$ because we assume $c_0 = \frac{1}{(1-\epsilon_0)\lambda}\big\{  7 + \frac{5}{(1-\epsilon_0)\lambda} \big\} > \frac{2}{(1-\epsilon)\lambda}(C_* +1 + o(1)) + c'$.
	
	It is easy to see that the maximum \eqref{M2_part2} over $k \in M_2$ is also of order $o(1)$ with probability tending to 1 by  the similar arguments.
	Since we have \eqref{sel_goal1} in the proof of Theorem \ref{thm:nosuper}, it completes the proof.
\end{proof}

\begin{proof}[Proof of Theorem \ref{thm:selection}]
	Let $M_2 = \{k : k \nsupseteq t,  |G_k|\le m_n \}$.
	Since we have Theorem \ref{thm:nosuper}, it suffices to show that 
	\bean\label{sel_cons_goal}
	\sum_{k: k \in M_2} PR(k,t) &\overset{P}{\lra}& 0 \quad \text{ as } n\to\infty.
	\eean	
	By the proof of Theorem \ref{thm:ratio}, the summation of \eqref{M2_part1} over $k \in M_2$ is bounded above by
	\bea
	&&\sum_{k\in M_2} \Big( \frac{2}{r} \Big)^{|k|-|t|} \exp \Big\{ -\Big( \frac{(1-\epsilon)\lambda}{2}  (c_0- c')  - C_* - o(1) \Big) |t\setminus k|^2 |G_t|  \Lambda_{|t|} \log r    \Big\}   \\
	&\le& \sum_{|k|=0}^r \sum_{v = 0}^{(|t|-1) \wedge |k|} \binom{|t|}{v} \binom{r-|t|}{|k|-v} \Big( \frac{2}{r} \Big)^{|k|-|t|}  \exp \Big\{ - \Big( \frac{(1-\epsilon)\lambda}{2}  (c_0- c')  - C_* - o(1) \Big)(|t|-v)^2 |G_t|  \Lambda_{|t|} \log r    \Big\} \\
	&\le& \sum_{|k|=0}^r \sum_{v = 0}^{(|t|-1) \wedge |k|}  \big( |t| r \big)^{|t| -v} \exp \Big\{ - \Big( \frac{(1-\epsilon)\lambda}{2}  (c_0- c')  - C_* - o(1) \Big) (|t|-v)^2 |G_t|  \Lambda_{|t|} \log r    \Big\} \\
	&\le& \exp \Big\{ - \Big( \frac{(1-\epsilon)\lambda}{2}  (c_0- c')  - C_* - o(1) \Big)  |G_t|  \Lambda_{|t|} \log r + 2(|t|+2) \log r    \Big\} \\
	&\le& \exp \Big\{ -   \Big( \frac{(1-\epsilon)\lambda}{2}  (c_0- c')  - C_* - 6 - o(1) \Big)  |G_t|  \Lambda_{|t|} \log r \Big\}
	\eea
	with probability tending to 1, where $C_* \le 2 +\delta$ is defined in the proof of Theorem \ref{thm:ratio}.
	Note that the last term is of order $o(1)$ because we assume $c_0 = \frac{1}{(1-\epsilon_0)\lambda}\big\{ 17 + \frac{5}{(1-\epsilon_0)\lambda} \big\} > \frac{2}{(1-\epsilon)\lambda}(C_* + 6 + o(1)) + c'$.
	It is easy to see that the summation of \eqref{M2_part2} over $k \in M_2$ is also of order $o(1)$ with probability tending to 1 by  the similar arguments.
\end{proof}

\begin{lemma}\label{lem_aux1}
	Under Condition \hyperref[cond_A2]{\rm (A2)}, we have
	\bea
	\exp \big\{ \frac{1}{2} \what{\beta}_{G_t}^T \big( A_t - A_t(A_t+\tau^{-2}I)^{-1} A_t \big)\what{\beta}_{G_t}   \big\}  
	&\lesssim& 1
	\eea
	for any $k \in M_1$ with probability tending to 1.	
\end{lemma}

\begin{proof}
	Note that, by Condition \hyperref[cond_A2]{\rm (A2)},
	\bea
	(A_t + \tau^{-2} I)^{-1} &\ge& (A_t + (n\tau^2 \lambda)^{-1} A_t)^{-1} \,\,=\,\, \frac{n\tau^2 \lambda}{n\tau^2 \lambda+1} A_t^{-1},
	\eea
	which implies that
	\bea
	\frac{1}{2} \what{\beta}_{G_t}^T \big( A_t - A_t(A_t+\tau^{-2}I)^{-1} A_t \big)\what{\beta}_{G_t}  
	&\le& \frac{1}{2(n\tau^2 \lambda+1)} \what{\beta}_{G_t}^T A_t \what{\beta}_{G_t}   .
	\eea
	Thus, we complete the proof if we show that
	\bea
	 \frac{1}{ n\tau^2 \lambda } \what{\beta}_{G_t}^T H_n(\beta_{0,G_t}) \what{\beta}_{G_t} 
	 &\le& C
	\eea
	for some constant $C>0$ and any $k\in M_1$ with probability tending to 1.
	By Lemma A.3 in \cite{Naveen:2018} and Condition \hyperref[cond_A2]{\rm (A2)},
	\bea
	 \frac{1}{n \tau^2 }\what{\beta}_{G_t}^T H_n(\beta_{0,G_t}) \what{\beta}_{G_t}
	 &\le& \frac{1}{\tau^2 } \lambda_{\max}\big\{ n^{-1} H_n(\beta_{0,G_t}) \big\} \|\what{\beta}_{G_t}\|_2^2 \\
	 &\le& \frac{1}{\tau^2} \Big( \frac{n}{\log r} \Big)^d \big( \|\beta_{0,G_t}\|_2^2  + o(1) \big) 
	 \,\,=\,\, O(1)
	\eea
	for any $k\in M_1$ with probability tending to 1.
\end{proof}

\begin{lemma}\label{lem_aux2}
	Under Conditions \hyperref[cond_A2]{\rm (A2)} and \hyperref[cond_A4]{\rm (A4)}, we have
	\bea
	DR(k, t) 
	&:=& \frac{\det \Big\{ n^{-1} H_n(\beta_{0,G_t}) (1+\epsilon) + (n\tau)^{-2}I  \Big\} }{\det \Big\{ n^{-1}H_n(\beta_{0,G_k}) (1-\epsilon) + (n\tau)^{-2}I  \Big\}} \\
	&\lesssim& \Big( \frac{2}{\lambda} \Big)^{|G_k|- |G_t|}
	\eea
	for any $k \in M_1$.	
\end{lemma}

\begin{proof}
	Since $\sg(\beta_{0,G_t}) = \sg(\beta_{0,G_k})$ for any $k\in M_1$,
	\bea
	&& \det \Big\{ n^{-1}H_n(\beta_{0,G_k}) (1-\epsilon) + (n\tau)^{-2}I  \Big\} \\
	&=& \det \Big\{ \frac{1-\epsilon}{n} X_{G_k}^T \sg(\beta_{0,G_t}) X_{G_k} + (n\tau)^{-2}I  \Big\} \\
	&=& \det (A) \, \det \Big\{ \frac{1-\epsilon}{n} X_{G_{k\setminus t}}^T \sg(\beta_{0,G_t}) X_{G_{k\setminus t}}  + (n\tau)^{-2}I - \frac{(1-\epsilon)^2}{n^2} X_{G_{k\setminus t}}^T \sg(\beta_{0,G_t}) X_{G_{t}} A^{-1}  X_{G_{t}}^T \sg(\beta_{0,G_t}) X_{G_{k\setminus t}} \Big\} \\
	&\equiv& \det (A) \det(R_{k\setminus t}),
	\eea
	where 
	\bea
	A &=& \frac{1-\epsilon}{n} X_{G_t}^T \sg(\beta_{0,G_t}) X_{G_t}  + (n\tau)^{-2}I 
	\eea
	by the inversion formula of the block matrix.
	Then, we have
	\bea
	DR(k,t) &=& \frac{ \det \Big\{ \frac{1+\epsilon}{n} X_{G_t}^T \sg(\beta_{0,G_t}) X_{G_t}  + (n\tau)^{-2}I  \Big\} }{ \det \Big\{ \frac{1-\epsilon}{n} X_{G_t}^T \sg(\beta_{0,G_t}) X_{G_t}  + (n\tau)^{-2}I  \Big\} } \det (R_{k\setminus t}^{-1}) .
	\eea
	
	Since $R_{k\setminus t}^{-1}$ is a principal submatrix of $\big\{ (1-\epsilon)n^{-1} H_n(\beta_{0,G_k}) + (n\tau^2)^{-1} I \big\}^{-1}$,
	\bea
	\lambda_{\max}(R_{k\setminus t}^{-1})
	&\le& \lambda_{\max}\Big[  \Big\{ \frac{1-\epsilon}{n} H_n(\beta_{0,G_k}) + \frac{1}{n\tau^2}I  \Big\}^{-1}  \Big] \\
	&=& \Big[ \lambda_{\min} \Big\{ \frac{1-\epsilon}{n} H_n(\beta_{0,G_k}) + \frac{1}{n\tau^2}I  \Big\}  \Big]^{-1}\\
	&\le& \frac{1}{1-\epsilon} \Big[  \lambda_{\min} \Big\{  \frac{1}{n}H_n(\beta_{0,G_k}) \Big\} \Big]^{-1} 
	\,\,\le\,\, \frac{1}{(1-\epsilon)\lambda} .
	\eea
	Thus, we have
	\bea
	\det (R_{k\setminus t}^{-1}) 
	&\le& \big\{ (1-\epsilon)\lambda \big\}^{-(|G_k| - |G_t|)} \,\, \le \,\, \Big( \frac{2}{\lambda} \Big)^{|G_k|- |G_t|}.
	\eea

	Now, we complete the proof if we show that
	\bea
	\frac{ \det \Big\{ \frac{1+\epsilon}{n} X_{G_t}^T \sg(\beta_{0,G_t}) X_{G_t}  + (n\tau)^{-2}I  \Big\} }{ \det \Big\{ \frac{1-\epsilon}{n} X_{G_t}^T \sg(\beta_{0,G_t}) X_{G_t}  + (n\tau)^{-2}I  \Big\} } 
	&=& O(1).
	\eea
	Let $\lambda_1, \ldots, \lambda_{|G_t|}$ be eigenvalues of $n^{-1} H_n(\beta_{0,G_t})$.
	Then
	\bea
	\frac{ \det \Big\{ \frac{1+\epsilon}{n} X_{G_t}^T \sg(\beta_{0,G_t}) X_{G_t}  + (n\tau)^{-2}I  \Big\} }{ \det \Big\{ \frac{1-\epsilon}{n} X_{G_t}^T \sg(\beta_{0,G_t}) X_{G_t}  + (n\tau)^{-2}I  \Big\} }  
	&=& \prod_{j=1}^{|G_t|} \frac{(1+\epsilon)\lambda_j + (n\tau^2)^{-1} }{(1-\epsilon)\lambda_j + (n\tau^2)^{-1}} \\
	&\le& \Big\{  \frac{1+\epsilon}{1-\epsilon} + \frac{1}{(1-\epsilon)\lambda n \tau^2}  \Big\}^{|G_t|}  \\
	&=& \Big(  1 + \frac{3\epsilon}{1-\epsilon}   \Big)^{|G_t|} ,
	\eea
	since $\epsilon =  c\sqrt{m_n^2 \Lambda_{m_n} \log r /n } $ and $\epsilon n\tau^2 \gtrsim 1$ by Conditions \hyperref[cond_A2]{\rm (A2)} and \hyperref[cond_A4]{\rm (A4)}.
	The last term is bounded above by
	\bea
	(  1 +4\epsilon )^{|G_t|} &\lesssim& \exp \big\{  4 \epsilon |G_t| \big\} \,\,= \,\, O(1),
	\eea
	by Condition \hyperref[cond_A2]{\rm (A2)}.
\end{proof}

\begin{lemma}\label{lem_aux3}
	Let $b_n = (1+\delta^*)(1+2w) \log r$ such that $1+\delta > (1+\delta^*)(1+2w)$. Then, under Condition \hyperref[cond_A2]{\rm (A2)},
	\bea
	L_n(\what{\beta}_{G_k}) - L_n(\what{\beta}_{G_t})  
	&\le& b_n (|G_k| - |G_t|) + C
	\eea
	for some constant $C>0$ and any $k \in M_1$ with probability tending to 1. 
\end{lemma}

\begin{proof}
	By Taylor's expansion and Lemma A.1 in \cite{Naveen:2018},
	\bea
	L_n(\what{\beta}_{G_k}) - L_n(\beta_{0,G_k})  &=& (\what{\beta}_{G_k} - \beta_{0,G_k})^T s_n(\beta_{0,G_k})  - \frac{1}{2} (\what{\beta}_{G_k} - \beta_{0,G_k})^T H_n(\tilde{\beta}_{G_k}) (\what{\beta}_{G_k} - \beta_{0,G_k}) \\
	&\le& (\what{\beta}_{G_k} - \beta_{0,G_k})^T s_n(\beta_{0,G_k})  - \frac{1-\epsilon}{2} (\what{\beta}_{G_k} - \beta_{0,G_k})^T H_n({\beta}_{0,G_k}) (\what{\beta}_{G_k} - \beta_{0,G_k}) 
	\eea
	for some $\|\tilde{\beta}_{G_k} - {\beta}_{0,G_k}\|_2 \le \| \what{\beta}_{G_k}  - {\beta}_{0,G_k}\|_2$.
	Note that by Lemmas A.1 and A.3 in \cite{Naveen:2018}, with probability tending to 1,
	\bea
	&& (\what{\beta}_{G_k} - \beta_{0,G_k})^T s_n(\beta_{0,G_k})  - \frac{1-\epsilon}{2} (\what{\beta}_{G_k} - \beta_{0,G_k})^T H_n({\beta}_{0,G_k}) (\what{\beta}_{G_k} - \beta_{0,G_k}) \\
	&\le& \frac{1}{2(1-\epsilon)} s_n(\beta_{0,G_k})  ^T  H_n({\beta}_{0,G_k})^{-1}   s_n(\beta_{0,G_k})  \\
	&=&  \frac{1}{2(1-\epsilon)} (E- \mu)^T X_{G_k} H_n({\beta}_{0,G_k})^{-1}  X_{G_k}^T (E-\mu)  \\
	&\equiv& \frac{1}{2(1-\epsilon)}  \tilde{U}^T  P_{G_k} \tilde{U}
	\eea
	for any $k \in M_1$, where $P_{G_k} = \sg^{1/2} X_{G_k} H_n(\beta_{0,G_k})^{-1} X_{G_k}^T \sg^{1/2}$, $\tilde{U} = \sg^{-1/2}(E-\mu)$ and $\sg = diag(\sigma_1^{2}, \ldots, \sigma_n^{2})$.
	Similarly, with probability tending to 1,
	\bea
	L_n(\what{\beta}_{G_t}) - L_n(\beta_{0,G_t}) 
	&\ge& \frac{1}{2(1+\epsilon)}  \tilde{U}^T  P_{G_t} \tilde{U}.
	\eea
	Therefore, with probability tending to 1,
	\bea
	L_n(\what{\beta}_{G_k}) - L_n(\what{\beta}_{G_t})  
	&=& L_n(\what{\beta}_{G_k}) - L_n(\beta_{0,G_k})   - \Big\{ L_n(\what{\beta}_{G_k}) - L_n(\beta_{0,G_k})   \Big\}  \\
	&\le& \frac{1}{2(1-\epsilon)} \tilde{U}^T \big( P_{G_k} - P_{G_t} \big) \tilde{U} + \frac{\epsilon}{(1-\epsilon)(1+\epsilon)} \tilde{U}^T P_{G_t} \tilde{U} 
	\eea
	for any $k \in M_1$.

	Note that, by Lemma A.2 in \cite{Naveen:2018} and Condition \hyperref[cond_A2]{\rm (A2)}, with probability tending to 1,
	\bea
	\frac{1}{2(1-\epsilon)} \tilde{U}^T \big( P_{G_k} - P_{G_t} \big) \tilde{U}  &\le& b_n(|G_k|-|G_t|)
	\eea
	for any $k \in M_1$. 
	Again by Lemma A.2 in \cite{Naveen:2018} and Condition \hyperref[cond_A2]{\rm (A2)}, we have
	\bea
	\frac{\epsilon}{(1-\epsilon)(1+\epsilon)} \tilde{U}^T P_{G_t} \tilde{U} 
	&=& O_p(1)
	\eea
	because $\epsilon |G_t| \le \epsilon m_n = O(1)$ and $\epsilon = o(1)$.	
\end{proof}

\bibliographystyle{plainnat} 
\bibliography{references}
\end{document}